\documentclass[11pt]{article}
\usepackage[margin=1.25in]{geometry}
\usepackage{ifthen}
\usepackage{graphicx}
\usepackage{epsfig}
\usepackage{dsfont}
\usepackage{amsfonts,dsfont,amssymb,amsthm,stmaryrd,bbm}
\usepackage[cmex10]{amsmath} 

\usepackage{hyperref}
\hypersetup{colorlinks=true,pdftitle="",pdftex}

\newtheorem{conj}{Conjecture}[section]
\newtheorem{thm}[conj]{Theorem}

\newtheorem{remark}[conj]{Remark}
\newtheorem{lem}[conj]{Lemma}
\newtheorem{prop}[conj]{Proposition}

\newtheorem{defn}[conj]{Definition}
\newtheorem{cor}[conj]{Corollary}

\newcommand\independent{\protect\mathpalette{\protect\independent}{\perp}} 
\def\independent#1#2{\mathrel{\rlap{$#1#2$}\mkern2mu{#1#2}}}

\newcommand{\supp}{\mathrm{Supp}}

\newcommand{\mR}{\mathbb{R}} 
\newcommand{\R}{\mathbb{R}}

\newcommand{\bc}{\begin{center}}
\newcommand{\ec}{\end{center}}
\newcommand{\bt}{\begin{tabular}}
\newcommand{\et}{\end{tabular}} 
\newcommand{\bea}{\begin{eqnarray}}
\newcommand{\eea}{\end{eqnarray}}
\newcommand{\bean}{\begin{eqnarray*}}
\newcommand{\eean}{\end{eqnarray*}}

\newcommand{\ba}{\begin{array}}
\newcommand{\ea}{\end{array}}

\def\be{\begin{eqnarray}}
\def\ee{\end{eqnarray}}
\def\ben{\begin{eqnarray*}}
\def\een{\end{eqnarray*}}

\newcommand{\ra} {\rightarrow}

\newcommand{\leqc}{\mbox{$ \;\stackrel{(c)}{\leq}\; $}}

\newcommand{\geqa}{\mbox{$ \;\stackrel{(a)}{\geq}\; $}}
\newcommand{\geqb}{\mbox{$ \;\stackrel{(b)}{\geq}\; $}}

\newcommand{\RL}{{\mathbb R}}

\newcommand{\calM}{\mbox{${\cal M}$}}

\newcommand{\Rpl}{{\mathbb R}_{+}}

\newcommand{\Nat}{\mathbb{N}}

\newcommand{\lam}{\lambda}

\def\elabel#1{\label{e:#1}}

\def\sq{$\Box$}

\def\qed{\ifmmode\sq\else{\unskip\nobreak\hfil
\penalty50\hskip1em\null\nobreak\hfil\sq
\parfillskip=0pt\finalhyphendemerits=0\endgraf}\fi\par\medbreak}

\def\supp{{\rm supp\,}}

\newsavebox{\junk}
\savebox{\junk}[1.6mm]{\hbox{$|\!|\!|$}}

\def\til={{\widetilde =}}

\def\half{{\mathchoice{\textstyle \frac{1}{2}}%
{\frac{1}{2}}%
{\hbox{\tiny $\frac{1}{2}$}}%
{\hbox{\tiny $\frac{1}{2}$}} }}

 \def\eq#1/{(\ref{#1})}

\def\eq#1/{(\ref{e:#1})}

\newcommand{\beqn}[1]{\notes{#1}%
\begin{eqnarray} \elabel{#1}}

\newcommand{\eeqn}{\end{eqnarray} }

\newcommand{\beq}[1]{\notes{#1}%
\begin{equation}\elabel{#1}}

\newcommand{\eeq}{\end{equation}} 

\def\bdes{\begin{description}}
\def\edes{\end{description}}

\def\notes#1{}

\newcommand{\setS}{s}

\newcommand{\collS}{\mathcal{C}}

\newcommand{\sumS}{\sum_{\setS\in\collS}}

\newcommand{\bs}{\beta_{\setS}}

\newcommand{\dth}{\frac{1}{d}}

\def\E{{\bf E}}

\def\phi{\varphi}

\def\bee{\begin{eqnarray*}}
\def\ene{\end{eqnarray*}}

\begin{document}

\title{Forward and Reverse Entropy Power Inequalities in Convex Geometry}
\author{Mokshay Madiman, James Melbourne, Peng Xu\thanks{All the authors are with the 
Department of Mathematical Sciences, University of Delaware.
E-mail: {\tt madiman@udel.edu, jamesm@udel.edu, xpeng@udel.edu}.
This work was supported in part by the U.S. National Science Foundation through grants DMS-1409504 (CAREER) and CCF-1346564. 
Some of the new results described in Section~\ref{sec:repi} were announced 
at the 2016 IEEE International Symposium on Information Theory \cite{XMM16:isit} in Barcelona.}}
\date{}
\maketitle

\begin{abstract}
The entropy power inequality, which plays a fundamental role in information theory and probability,
may be seen as an analogue of the Brunn-Minkowski inequality. Motivated by this connection to 
Convex Geometry, we survey various recent developments  on forward and reverse entropy power inequalities 
not just for the Shannon-Boltzmann entropy but also more generally for R\'enyi entropy.
In the process, we discuss connections between the so-called functional (or integral) and probabilistic (or entropic) 
analogues of some classical inequalities in geometric functional analysis.
\end{abstract}

\tableofcontents

\section{Introduction}

The Brunn-Minkowski inequality plays a fundamental role not just in Convex Geometry,
where it originated over 125 years ago, but also as an indispensable tool in Functional Analysis,
and-- via its connections to the concentration of measure phenomenon-- in
Probability. The importance of this inequality, and the web of its tangled relationships
with many other interesting and important inequalities, is beautifully elucidated in the
landmark 2002 survey of Gardner \cite{Gar02}. Two of the parallels that Gardner discusses
in his survey are the Pr\'ekopa-Leindler inequality and the Entropy Power Inequality;
since the time that the survey was written, these two inequalities have become the foundation
and prototypes for two different but related analytic ``liftings'' of Convex Geometry. 
While the resulting literature is too vast for us to attempt doing full justice to in this survey,
we focus on one particular strain of research-- namely, the development of reverse
entropy power inequalities-- and using that as a narrative thread, chart some of the work
that has been done towards these ``liftings''.

Let  $A, B$ be any nonempty Borel sets in $\R^d$. Write  $A + B = \big\{x+y: x \in A, \ y \in B\big\}$
for the Minkowski sum, and  $|A|$ for the $d$-dimensional volume (or Lebesgue measure) of $A$.
The Brunn-Minkowski inequality (BMI) says that 
\be\label{eq:bm-std}
\big| A +  B\big|^{1/d} \geq |A|^{1/d} + |B|^{1/d} .
\ee
The BMI was proved in the late 19th century by Brunn for convex sets in low dimension ($d \leq 3$), and Minkowski for convex sets in $\R^d$;
the reader may consult Kjeldsen \cite{Kje08, Kje09} for an interesting historical analysis of how the notion of convex sets in linear spaces
emerged from these efforts (Minkowski's in particular). The extension of the BMI to compact-- and thence Borel-measurable-- 
subsets of $\R^d$ was done by Lusternik \cite{Lus35}. Equality holds in the inequality \eqref{eq:bm-std}
for sets $A$ and $B$ with positive volumes if and only if they are convex and homothetic (i.e., one is a scalar multiple of the other, up to translation),
possibly with sets of measure zero removed from each one.
As of today, there are a number of simple and elegant proofs known for the BMI. 

In the last few decades, the BMI became the starting point of what is sometimes called the Brunn-Minkowski theory,
which encompasses a large and growing range of geometric inequalities including the
Alexandrov-Fenchel inequalities for mixed volumes, and which has 
even developed important offshoots such as the $L^p$-Brunn-Minkowski theory \cite{Lut93:1}.
Already in the study of the geometry of convex bodies (i.e., convex compact sets with nonempty interior), 
the study of log-concave functions turns out to be fundamental.
One way to see this is to observe that  uniform measures
on convex bodies are not closed under taking lower-dimensional marginals, but yield
log-concave densities, which do have such a closure property-- while the closure property of
log-concave functions under marginalization goes back to Pr\'ekopa \cite{Pre71, Pre73} and Brascamp-Lieb \cite{BL76a},   
their consequent fundamental role in the geometry of convex bodies was first clearly recognized in
the doctoral work of K.~Ball \cite{Bal86:phd} (see also \cite{Bal88, MP89}). 
Since then, the realization has grown that it is both possible and natural to state many questions and theorems in Convex Geometry directly
for the category of log-concave functions or measures rather than for the category of convex bodies-- V.~Milman calls this
the ``Geometrization of Probability'' program \cite{Mil08}, although one might equally well call it the
``Probabilitization of Convex Geometry'' program. The present survey squarely falls within
this program.

For the goal of embedding the geometry of  convex sets  in a more analytic setting,
two approaches are possible:
\begin{enumerate}
\item {\it Functional (integral) lifting}: Replace sets by functions, and convex sets by log-concave or $s$-concave functions,
and the volume functional  by the integral. This is a natural extension because if we identify a convex body $K$ with its
indicator function $1_K$ (defined as being 1 on the set and 0 on its complement), then the integral of $1_K$ is just
the volume of $K$. The earlier survey of V.~Milman \cite{Mil08} is entirely focused on this lifting of Convex Geometry;
recent developments since then include the introduction and study of mixed integrals (analogous to mixed volumes)
independently by Milman-Rotem \cite{MR13:1, MR13:2} and Bobkov-Colesanti-Fragala \cite{BCF14} (see also \cite{BGW13}). 
Colesanti \cite{Col16} has an up-to-date survey of these developments in another chapter of this volume. 
\item {\it Probabilistic (entropic) lifting}:  Replace sets by random variables (or strictly speaking their distributions), 
and convex sets by random variables with log-concave or $s$-concave distributions, and the volume functional 
by the entropy functional (actually ``entropy power'', which we will discuss shortly). This is a natural analogue 
because if we identify a convex body $K$ with the random variable $U_K$ whose distribution is uniform measure on $K$, then
the entropy of $U_K$ is the logarithm of $|K|$.
The  parallels were observed early by Costa and Cover \cite{CC84} (and perhaps also implicitly by Lieb \cite{Lie78}); 
subsequently this analogy has been studied by many other authors, including by  Dembo-Cover-Thomas \cite{DCT91} 
and in two series of papers by Lutwak-Yang-Zhang (see, e.g., \cite{LYZ04, LLYZ13}) and Bobkov-Madiman (see, e.g., \cite{BM11:cras, BM12:jfa}).
\end{enumerate}

While this paper is largely focused on the probabilistic (entropic) lifting, we will also discuss how it is related 
to the functional (integral) lifting.

It is instructive at this point to state the integral and entropic liftings of the Brunn-Minkowski inequality itself,
which are known as the Pr\'ekopa-Leindler inequality and the Entropy Power Inequality respectively.

\vspace{.1in}
\noindent{\bf Pr\'ekopa-Leindler inequality (PLI)}:
The Pr\'ekopa-Leindler inequality (PLI) \cite{Pre71, Lei72b, Pre73} states that if $f, g, h:\R^d\ra [0,\infty)$ are integrable functions satisfying, for a given $\lam\in (0,1)$,
\ben
h(\lam x + (1-\lam) y) \geq f^\lam(x) g^{1-\lam}(y)
\een
for every $x, y\in \R^d$, then
\be\label{eq:pli}
\int h \geq \bigg(\int f\bigg)^\lam \bigg(\int g\bigg)^{1-\lam} .
\ee
If one prefers, the PLI can also be written more explicitly as a kind of convolution inequality, as implictly observed
in \cite{BL76a} and explicitly in \cite{KM05}. Indeed, if one defines the Asplund product of two nonnegative functions by
\ben
    (f \star g) (x) = \sup_{x_1 + x_2 = x} f(x_1) g(x_2) ,
\een
and the scaling $(\lambda \cdot f) (x) = f^\lambda(x/\lambda)$,
then the left side of \eqref{eq:pli} can be replaced by the integral of $[\lam\cdot f]\star [(1-\lam)\cdot g]$.

To see the connection with the BMI, one simply has to observe that $f=1_A, g=1_B$
and $h=1_{\lam A + (1-\lam)B}$ satisfy the hypothesis, and in this case, the conclusion is precisely the BMI in its
``geometric mean'' form
$|\lam A +  (1-\lam)B| \geq |A|^{\lam} |B|^{1-\lam}$.
The equivalence of this inequality to the BMI in the form \eqref{eq:bm-std} 
is just one aspect of a broader set of equivalences involving the BMI. To be precise, 
for the class of Borel-measurable subsets of $\mathbb{R}^d$, the following are equivalent:
\begin{align}
	|A+B|^{\frac 1 d} 
		&\geq 	|A|^{\frac 1 d} + |B|^{\frac 1 d} \label{trad}
			\\
	|\lam A + (1-\lam)B| 
		&\geq \left( \lam |A|^{\frac 1 d} + (1-\lam) |B|^{\frac 1 d} \right)^d \label{sharp}
			\\
	|\lam A + (1-\lam)B|
		&\geq
			|A|^{\lam}|B|^{1-\lam} \label{geo}
			\\
	|\lam A + (1-\lam)B|
		&\geq
			\min \{ |A|, |B| \}. \label{min}
\end{align} 

Let us indicate why the inequalities \eqref{trad}--\eqref{min} are equivalent.
Making use of the arithmetic mean-geometric mean inequality, we immediately have \eqref{sharp} $\Rightarrow$ \eqref{geo} $\Rightarrow$ \eqref{min}.  Applying \eqref{trad} to $\tilde{A} = \lam A$, $\tilde{B} = (1-\lam)B$ we have
\begin{align*}
|\lam A + (1-\lam)B| 
	&=
		|\tilde{A} + \tilde{B}|
		\\	
	&\geq 
		( |\tilde{A}|^{\frac 1 d} + |\tilde{B}|^{\frac 1 d} )^d
		\\
	&=  
		\left( |\lam A|^{\frac 1 d} + |(1-\lam)B|^{\frac 1 d} \right)^d
		\\
	&=
		\left( \lam |A|^{\frac 1 d} + (1-\lam)|B|^{\frac 1 d} \right)^d,
\end{align*}
where the last equality is by homogeneity of the Lebesgue measure.  Thus \eqref{trad} $\Rightarrow$ \eqref{sharp}.  It remains to prove that \eqref{min} $\Rightarrow$ \eqref{trad}.  First notice that \eqref{min} is equivalent to 
\begin{align*}
	|A+B| 
		&\geq 
			\min \{|A/\lam|, |B/(1-\lam)| \}
			\\
		&=
			\min \{ |A|/\lam^d , |B|/(1-\lam)^d \}.
\end{align*}
It is easy to see that the right hand side is maximized when $|A|/\lam^d = |B|/(1-\lam)^d$, or
\[
	\lam = \frac{ |A|^{\frac 1 d} }{|A|^{\frac 1 d} +  |B|^{\frac 1 d}}.
\]
Inserting $\lam$ into the above yields \eqref{trad}.

\vspace{.1in}
\noindent{\bf Entropy Power Inequality (EPI)}:
In order to state the Entropy Power Inequality (EPI), let us first explain what is meant by entropy power.
When random variable $X=(X_1,\ldots,X_d)$ has density $f(x)$ on $\RL^d$, the {\it entropy} of $X$ is 
\be\label{eq:SNE}
h(X)=h(f):= -\int_{\RL^d} f(x)\log f(x)dx= \E[-\log f(X)] .
\ee
This quantity is sometimes called the Shannon-Boltzmann entropy or the differential entropy (to distinguish
it from the discrete entropy functional that applies to probability distributions on a countable set).
The {\it entropy power} of $X$ is 
$
N(X)=e^{\frac{2h(X)}{d}}.
$
As is usual, we abuse notation and write $h(X)$ and $N(X)$,
even though these are functionals depending only on the density
of $X$ and not on its random realization.
The entropy power $N(X)\in [0,\infty]$ can be thought of as a ``measure of randomness''. 
It is an (inexact) analogue of volume: 
if $U_A$ is uniformly distributed on a bounded Borel set $A$, then it is easily checked that 
$h(U_A)=\log |A|$ and hence $N(U_A)=|A|^{2/d}$.
The reason we don't define entropy power by $e^{h(X)}$ (which would yield a value of $|A|$ for the entropy power of $U_A$)
is that the ``correct'' comparison is not to uniforms but to Gaussians.
This is because just as Euclidean balls are special among subsets of $\RL^d$, 
Gaussians are special among distributions on $\RL^d$.
Indeed, the reason for the appearance of the  functional $|A|^{\dth}$ 
in the BMI is because this functional is (up to a universal constant)
the radius of the ball that has the same volume as $A$, i.e., $|A|^{\dth}$ may be thought of as
(up to a universal constant) the ``effective radius'' of $A$.
To develop the analogy for random variables, observe that when $Z\sim N(0,\sigma^{2}I)$
(i.e., $Z$ has the Gaussian distribution with mean 0 and covariance matrix that is a multiple of the identity), 
the entropy power of $Z$ is $N(Z)=(2\pi e) \sigma^{2}$.
Thus the entropy power of $X$ is (up to a universal constant)
the variance of the isotropic normal that has the same entropy as $X$,
i.e., if $Z\sim N(0,\sigma_Z^{2}I)$ and $h(Z)=h(X)$, then
\ben
N(X)=N(Z)=(2\pi e) \sigma_{Z}^{2} .
\een
Looked at this way, entropy power is the ``effective variance'' of a random variable,
exactly as volume raised to $1/d$ is the effective radius of a set.

The EPI states that for any two independent random vectors $X$ and $Y$ in $\R^d$ such that the entropies of $X, Y$ and $X+Y$ exist, 
\ben
N(X+Y) \geq N(X) + N(Y) .
\een
The EPI was stated by Shannon \cite{Sha48} with an incomplete proof; the first complete proof
was provided by Stam \cite{Sta59}. The EPI plays an extremely important role in the field
of Information Theory, where it first arose and was used (first by Shannon, and later by many others)
to prove statements about the fundamental limits of communication over various models
of communication channels. Subsequently it has also been recognized as an extremely useful
inequality in Probability Theory, with close connections to the logarithmic Sobolev
inequality for the Gaussian distribution as well as to the Central Limit Theorem. 
We will not further discuss these other motivations for the study of the EPI
in this paper, although we refer the interested reader to \cite{Joh04:book, Mad17} for more
on the connections to central limit theorems.

It should be noted that one insightful way to compare the BMI and EPI is to think of the latter
as a ``99\% analogue in high dimensions'' of the former, in the sense that looking at 
most of the Minkowski sum of the supports of a large number of independent copies of the two 
random vectors effectively yields the EPI via a simple instance of 
the asymptotic equipartition property or Shannon-McMillan-Breiman theorem.
A rigorous argument is given by Szarek and Voiculescu \cite{SV00} (building on \cite{SV96}),
a short intuitive explanation of which can be found in an answer of Tao to a MathOverflow question\footnote{See 
{\tt http://mathoverflow.net/questions/167951/entropy-proof-of-brunn-minkowski-inequality.}}.
The key idea of \cite{SV00} is to use not the usual BMI but a ``restricted'' version of it where it is
the exponent $2/d$ rather than $1/d$ that shows up\footnote{We mention in passing that Barthe \cite{Bar99} also proved a restricted version of the PLI.
An analogue of ``restriction'' for the EPI would involve some kind of weak dependence between summands;
some references to the literature on this topic are given later.}.

\vspace{.1in}
\noindent{\bf R\'enyi entropies.}
Unified proofs can be given of the EPI and the BMI in two different ways, both of which
may be thought of as providing extensions of the EPI to R\'enyi entropy. We will discuss both of
these later; for now, we only introduce the notion of  R\'enyi entropy.
For a $\R^d$-valued random variable $X$ with probability density function $f$,
define its R\'enyi entropy of order $p$ (or simply $p$-R\'enyi entropy) by
\begin{eqnarray}\label{Def:Renyi}
h_p(X)=h_p(f):=\frac{1}{1-p}\log\left(\int_{\mathbb{R}^d}f^p(x)dx\right) ,
\end{eqnarray}
if $p\in (0,1)\cup (1,\infty)$.
Observe that, defining $h_1$ ``by continuity'' and using l'Hospital's rule, $h_1(X)=h(X)$ is the (Shannon-Boltzmann) entropy.
Moreover, by taking limits,
\ben\begin{split}
h_0(X) &=  \log |\supp(f)| ,\\
h_\infty(X) &= -\log \|f\|_\infty ,
\end{split}\een
where $\supp(f)$ is the support of $f$ (i.e., the smallest closed set such that $f$ is zero outside it),
and $\|f\|_\infty$ is the usual $L^\infty$-norm of $f$ (i.e., the essential supremum with respect to Lebesgue measure).
We also define the $p$-R\'enyi entropy power by $N_p(X)=e^{\frac{2h_p(X)}{d}}$, so that the
usual entropy power $N(X)=N_1(X)$ and for a random variable $X$ whose support is $A$, $N_0(X)=|A|^{2/d}$.

\vspace{.1in}
\noindent{\bf Conventions.}
Throughout this paper, {\it we assume that all random variables considered have densities with respect to Lebesgue measure}.
While the entropy of $X$ can be meaningfully set to $-\infty$ when the distribution of $X$ does not possess a density, for the most part
we avoid discussing this case. Also, when $X$ has probability density function $f$, we write $X\sim f$. 

For real-valued functions $A,B$ we will use the notation $A \lesssim B$ when $A(z) \leq C B(z)$ 
for some positive constant $C$ independent of $z$.  For our purposes this will be most interesting 
when $A$ and $B$ are in some way determined by dimension.

\vspace{.1in}
\noindent{\bf Organization.}
This survey is organized as follows. In Section~\ref{sec:epi}, we review various statements
and variants of the EPI, first for the usual Shannon-Boltzmann entropy in Section~\ref{sec:epi-s}
and then for $p$-R\'enyi entropy in Section~\ref{sec:epi-r}, focusing on the $\infty$-R\'enyi
case in Section~\ref{sec:epi-inf}. In Section~\ref{sec:repi}, we explore
what can be said about inequalities that go the other way, under convexity constraints on the probability measures
involved. We start by recalling the notions of $\kappa$-concave measures and functions in Section~\ref{sec:cvx}.
In Section~\ref{sec:repi-pos}, we discuss reverse EPI's that require invoking a linear transformation (analogous to the reverse Brunn-Minkowski
inequality of V.~Milman), and explicit choices of linear transformations that can be used are discussed in Section~\ref{ss:spl-pos}.
The three intermediate subsections focus on three different approaches to reverse
R\'enyi EPI's that do not require invoking a linear transformation.
Finally we discuss the relationship between integral and entropic liftings, in the context of the Blashke-Santal\'o inequality
in Section~\ref{sec:reln}, and end with some concluding remarks on nonlinear and discrete analogs in Section~\ref{sec:concl}.

\section{Entropy Power Inequalities}
\label{sec:epi}

\subsection{Some Basic Observations}
\label{sec:renyi}

Before we discuss more sophisticated results, let us recall some basic properties of R\'enyi entropy.

\begin{thm}\label{thm:CdtEntp}
For independent $\R^d$-valued random variables $X$ and $Y$, and any $p \in [0,\infty]$,
\begin{align*}
N_p(X+Y) \geq \max \{N_p(X), N_p(Y)\}.
\end{align*}
\end{thm}
\begin{proof}
Let $X\sim f$ and $Y\sim g$. For $p \in (1, \infty)$, we have the following with the inequality delivered by Jensen's inequality:
\begin{align*}
\int (f * g )^p(x) dx 
	&=
		\int \left( \int f(x-y) g(y) dy \right)^p dx
		\\
	&\leq
		\int \int f^p(x-y) g(y) dy dx
		\\
	&=
		\int \left( \int f^p(x-y) dx\right)  g(y) dy
		\\
	&=
		\int f^p(x) dx.
\end{align*}
Inserting the inequality into the order reversing function $\phi(z) = z^{\frac 2 {d(1-p)}}$ we have our result.

The case that $p \in (0,1)$ is similar, making note that now $z^p$ is concave while $z^{\frac 2 {d(1-p)}}$ is order preserving. 
For $p=1$, we can give a probabilistic proof: applying the nonnegativity of mutual information, which in particular
implies that conditioning reduces entropy (see, e.g., \cite{CT06:book}),
\ben
h(X+Y)\ge h(X+Y|Y)=h(X|Y)=h(X) ,
\een
where we used translation-invariance of entropy for the first equality and independence of $X$ and $Y$ for the second.
For $p=0$, the conclusion simply follows by the fact that
$|A+B|\ge \max \{|A|,|B| \}$ for any nonempty Borel sets $A$ and $B$; this may be seen by translating $B$ so that 
it contains 0, which does not affect any of the volumes and in which case $A+B\supset A$. 
For $p=\infty$, the conclusion follows from H\"older's inequality:
\begin{align*}
\int f(x-y)g(y)dy\le \|g\|_\infty \| f \|_1=\|g\|_\infty .
\end{align*}
Thus we have the theorem for all values of $p\in [0,\infty]$.
\end{proof}

We now observe that for any fixed random vector, the R\'enyi entropy of order $p$ is non-increasing in $p$.

\begin{lem}\label{lem:monotonicity}
For a $\mathbb{R}^d$-valued random variable $X$, and $0 \leq q < p \leq \infty$, we have
\begin{align*}
	N_q(X) \ge N_p(X).
\end{align*}
\end{lem}
\begin{proof}
The result follows by expressing, for $X \sim f$,
\[
h_p(X) = \frac{\log( \int f^p )}{1-p} = -\log \mathbb{E} \|f(X) \|_{p-1}
\]
and using the ``increasingness" of $p$-norms on probability spaces,
which is nothing but an instance of H\"older's inequality.
\end{proof}

\begin{defn}\label{defn:lc-fn}
A function $f: \R^d\ra [0,\infty)$
is said to be {\it log-concave} if 
\be\label{defn:lc}
f(\alpha x +(1-\alpha)y) \geq f(x)^{\alpha} f(y)^{1-\alpha} ,
\ee 
for each $x,y\in \R^d$ and each $0\leq \alpha\leq 1$. 
\end{defn}

If a probability density function $f$ is log-concave, we will
also use the adjective ``log-concave'' for a random variable $X$ distributed according
to $f$, and for the probability measure induced by it. 
Log-concavity has been deeply studied in probability, statistics, optimization and geometry, 
and is perhaps the most natural notion of convexity for probability density functions.

In general, the monotonicity of Lemma~\ref{lem:monotonicity} relates two different
R\'enyi entropies of the same distribution in one direction, but there is no reason for a 
bound to exist in the other direction.
Remarkably, for log-concave random vectors, all R\'enyi entropies are comparable
in both directions.

\begin{lem}\label{lem:theLemma} \cite{MWB16})
If a random variable $X$ in $\R^d$ has log-concave density $f$, then
for $p\geq q>0$,
\bee
h_q(f)-h_p(f)\leq d\frac{\log q}{q-1}-d\frac{\log p}{p-1} ,
\ene
with equality if $f(x)=e^{-\sum_{i=1}^d x_i}$ on the positive orthant and 0 elsewhere.
\end{lem}

This lemma generalizes  the following sharp inequality for log-concave distributions obtained in \cite{BM11:it}:
\be\label{inq:lem}
h(X)\le d+h_\infty(X) .
\ee

In fact, Lemma~\ref{lem:theLemma} has an extension to the larger class (discussed later) of $s$-concave measures
with $s<0$; preliminary results in this direction are available in \cite{BM11:it} and sharp results
obtained in \cite{BFLM16}.

\subsection{The Shannon-Stam EPI and its variants}
\label{sec:epi-s}

\subsubsection{The Basic EPI}

The  EPI has several  equivalent formulations; we collect these together with minimal conditions below.

\begin{thm}\label{thm:epi}
Suppose $X$ and $Y$ are independent $\R^d$-valued random variables such that $h(X), h(Y)$ and $h(X+Y)$ exist.
Then the following statements, which are equivalent to each other, are true:
\begin{enumerate}
\item 
We have
\be\label{inq:BasicEPIeqv} 
N(X+Y) \geq N(X) + N(Y) . \label{eq:BasicEPIeqv}
\ee
\item 
For any $\lambda\in [0,1]$,
\be\label{inq:eqv1}
h(\sqrt{\lambda}X+\sqrt{1-\lambda}Y)\ge \lambda h(X)+(1-\lambda)h(Y)\label{eq:eqv1}.
\ee
\item 
Denoting by $X^G$ and $Y^G$  independent, isotropic\footnote{By isotropic here, we mean spherical symmetry, or equivalently, that the covariance matrix is taken to be a scalar multiple of the identity matrix.}, Gaussian random variables with $h(X^G)=h(X)$ and $h(Y^G)=h(Y)$, one has
\be\label{inq:eqv2}
h(X+Y)\ge h(X^G+Y^G) .\label{eq:eqv2}
\ee
\end{enumerate}
In each case, equality holds if and only if $X$ and $Y$ are Gaussian random variables with proportional covariance matrices.
\end{thm}

\begin{proof}
First let us show that we can assume $h(X),h(Y) \in (-\infty, \infty)$.  
By Theorem \ref{thm:CdtEntp} we can immediately
obtain $h(X+Y) \geq \max \{h(X),h(Y) \}$.    It follows that all three inequalities hold immediately in the case that $\max\{h(X), h(Y)\} = \infty$. 
Now assume that neither $h(X)$ nor $h(Y)$ take the value $+\infty$ and consider $\min\{h(X),h(Y)\} = -\infty$.  In this situation, 
the inequalities \eqref{inq:BasicEPIeqv} and \eqref{inq:eqv1} are immediate.  For \eqref{inq:eqv2}, in the case that $h(X) = -\infty$ we interpret $X^G$ as a Dirac point mass, and hence $h(X^G+Y^G) = h(Y^G) = h(Y) \leq h(X+Y)$.  

We now proceed to prove the equivalences.

\eqref{inq:BasicEPIeqv} $\Rightarrow$ \eqref{inq:eqv1}: 
Apply \eqref{inq:BasicEPIeqv}, substituting $X$ by $ \sqrt{\lambda}X$ and $Y$ by $\sqrt{1-\lambda}Y$ and use the homogeneity of entropy power to obtain
\ben
N(\sqrt{\lambda}X+\sqrt{1-\lambda}Y)\ge \lambda N(X) + (1-\lambda)N(Y).
\een
Apply the AM-GM inequality to the right hand side and conclude by taking logarithms.

\eqref{inq:eqv1} $\Rightarrow$ \eqref{inq:eqv2}:
Applying \eqref{inq:eqv1} in its exponentiated form $N(\sqrt{\lambda}X + \sqrt{1-\lambda}Y) \geq N^{\lambda}(X) N^{1-\lambda}(Y)$ after writing $X + Y = \sqrt{\lambda} (X/\sqrt{\lambda}) +  \sqrt{1-\lambda} (Y/\sqrt{1-\lambda})$ we obtain
\ben
N(X+Y)\ge \left(N\left(\frac{X}{\sqrt{\lambda}}\right)\right)^{\lambda}\left(N\left(\frac{Y}{\sqrt{1-\lambda}}\right)\right)^{1-\lambda}.
\een
Making use of the identity $N(X^G + Y^G) = N(X^G)+N(Y^G)$ and homogeneity again, we can evaluate the right hand side at $\lambda =N(X^G)/N(X^G+Y^G)$ to obtain exactly $N(X^G+Y^G)$, recovering the exponentiated version of \eqref{inq:eqv2}.\\

\eqref{inq:eqv2} $\Rightarrow$ \eqref{inq:BasicEPIeqv}: Using the exponentiated version of \eqref{inq:eqv2},
\ben
N(X+Y)\ge N(X^G+Y^G)=N(X^G)+N(Y^G)=N(X)+N(Y).
\een
Observe from the proof that a strict inequality in one statement implies a strict inequality in the rest.

What is left is to prove any of the 3 statements of the EPI when the entropies involved are finite.
There are many proofs of this available in the literature (see, e.g., \cite{Sta59, Bla65, Lie78, DCT91, SV00, Rio11}), 
and we will not detail any here, although we later sketch a proof via the sharp form of Young's convolution inequality.
\end{proof}

The conditions stated above cannot be relaxed, as observed by Bobkov and Chistyakov \cite{BC15:1},
who construct a distribution whose entropy exists but such that the entropy of the self-convolution does not exist. 
This, in particular, shows that the assumption for validity of the EPI stated for example in \cite{DCT91}
is incomplete-- existence of just $h(X)$ and $h(Y)$ is not sufficient. 
It is also shown in \cite{BC15:1}, however, that for any example where $h(X)$ exists but $h(X+X')$
does not (with $X'$ an i.i.d. copy of $X$), necessarily $h(X)=-\infty$, so that
it remains true that if the entropy exists and is a real number, then the entropy of the self-convolution
also exists.
They also have other interesting examples of the behavior of entropy on convolution:
\cite[Example 1 ]{BC15:1} constructs a distribution with entropy $-\infty$ such that
that the entropy of the self-convolution is a real number, and \cite[Proposition 5]{BC15:1}  constructs a distribution with finite entropy
such that its convolution with any distribution of finite entropy has infinite entropy.

\subsubsection{Fancier versions of the EPI}
\label{sss:fancy-epi}

Many generalizations and improvements of the EPI exist.
For three or more independent random vectors $X_i$, the EPI trivially implies that 
\begin{eqnarray}\label{inq:MULSNEPI}
N(X_1+\cdots +X_n)\ge \sum_{i=1}^nN(X_i) ,
\end{eqnarray}
with equality if and only if the random vectors are Gaussian and their covariance matrices are proportional to each other. 
In fact, it turns out that this can be refined, as shown by S.~Artstein, K.~Ball, Barthe and Naor \cite{ABBN04:1}:
\begin{eqnarray}\label{inq:RFMEPI}
N\left(\sum_{i=1}^{n}X_i\right)\ge \frac{1}{n-1}\sum_{j=1}^{n}N\bigg(\sum_{i\neq j}X_i\bigg) .
\end{eqnarray}
This implies the monotonicity of entropy in the Central Limit Theorem, which suggests that quantifying the
Central Limit Theorem using entropy or relative entropy is a particularly natural approach. More precisely,
if $X_1,\ldots, X_n$ are independent and identically distributed (i.i.d.) square-integrable random vectors, then
\begin{eqnarray}\label{inq:ECLT}
h\left(\frac{X_1+\cdots +X_n}{\sqrt{n}}\right)\le h\left(\frac{X_1+\cdots +X_{n-1}}{\sqrt{n-1}}\right).
\end{eqnarray}
Simpler proofs of \eqref{inq:RFMEPI} were given independently by \cite{MB06:isit, SS07, TV06}. Generalizations of 
\eqref{inq:RFMEPI} to arbitrary collections of subsets on the right side was given by \cite{MB07, MG09:isit},
and some further fine properties of the kinds of inequalities that hold for the entropy power of a sum
of independent random variables were revealed in \cite{MG16}. Let us mention a key result of this type due to \cite{MG09:isit}.
For a collection  $\collS$ of nonempty subsets of $[n]:=\{1,2,\cdots,n\}$, a function $\beta:\collS \to \Rpl$ 
is called a {\em fractional partition}\footnote{If there exists
a fractional partition $\beta$ for $\collS$ that is $\{0,1\}$-valued,
then $\beta$ is the indicator function for a partition of the set $[n]$
using a subset of $\collS$; hence the terminology.} if  for each $i\in [n]$, we have
$\sum_{\setS\in \collS:i\in \setS} \bs = 1$. Then the entropy power of 
convolutions is fractionally superadditive, i.e., if $X_1, \ldots, X_n$ are independent
$\R^d$-valued random variables, one has
\ben
N\bigg(\sum_{i=1}^n X_i\bigg) \geq \sumS \bs N\bigg(\sum_{i\in\setS} X_i\bigg) .
\een
This yields the usual EPI by taking $\collS$ to be the collection of all singletons
and $\bs\equiv 1$, and the inequality \eqref{inq:RFMEPI} by taking $\collS$
to be the collection of all sets of size $n-1$ and $\bs\equiv \frac{1}{n-1}$.

For i.i.d. summands in dimension 1, \cite{ABBN04:2} and \cite{JB04} prove an upper bound of the relative entropy 
between the distribution of the normalized sum and that of a standard Gaussian random variable. 
To be precise, suppose $X_1, \ldots, X_n$ are independent copies of a random variable $X$ with $\mbox{Var}(X)=1$,
and the density of $X$ satisfies a Poincar\'e inequality with constant $c$, i.e., for every smooth function $s$,
\ben
c \mbox{Var}(s(X)) \leq \E [\{s'(X)\}^2] .
\een
Then, for every $a\in \mathbb{R}^n$ with $\sum_{i=1}^na_i^2=1$ and $\alpha(a):=\sum_{i=1}^na_i^4$,
\begin{eqnarray}\label{inq:upbdRltvEntp}
h(G)-h\left(\sum_{i=1}^na_iX_i\right)\le \frac{\alpha(a)}{\frac{c}{2}+(1-\frac{c}{2})\alpha(a)}\left(h(G)-h(X)\right),
\end{eqnarray} 
where $G$ is a standard Gaussian random variable.
Observe that this refines the EPI since taking $c=0$ in the inequality \eqref{inq:upbdRltvEntp}
gives the EPI in the second form of Theorem~\ref{thm:epi}.
On the other hand, specializing \eqref{inq:upbdRltvEntp} to $n=2$ with $a_1=a_2= \frac{1}{\sqrt{2}}$, one obtains
a lower bound for $h\big( \frac{X_1+X_2}{\sqrt{2}} \big) - h(X)$ in terms of the relative entropy $h(G)-h(X)$
of $X$ from Gaussianity. Ball and Nguyen \cite{BN12} develop an extension of this latter inequality to general dimension 
under the additional assumption of log-concavity.

It is natural to ask if the EPI can be refined by introducing an error term that quantifies the gap
between the two sides in terms of how non-Gaussian the summands are.
Such estimates are referred to as ``stability estimates'' since they capture how stable the equality
condition for the inequality is, i.e., whether closeness to Gaussianity is guaranteed for the summands
if the two sides in the inequality are not exactly equal but close to each other. 
For the EPI, the first stability estimates were given by Carlen and Soffer \cite{CS91}, but 
these are qualitative and not quantitative (i.e., they do not give numerical bounds on distance
from Gaussianity of the summands when there is near-equality in the EPI, but they do assert
that this distance must go to zero as the deficit in the inequality goes to zero).
Recently Toscani~\cite{Tos15:1} gave a quantitative stability estimate when the summands are restricted to have
log-concave densities: For independent random vectors $X$ and $Y$ with log-concave densities,
\begin{eqnarray}\label{inq:TosImprEPI}
N(X+Y)\ge \left(N(X)+N(Y)\right)R(X,Y),
\end{eqnarray}
where the quantity $R(X,Y)\ge 1$ is a somewhat complicated quantity that we do not define here and can be interpreted as a 
measure of non-Gaussianity of $X$ and $Y$. Indeed, \cite{Tos15:1} shows that $R(X,Y)=1$ if and only if $X$ and 
$Y$ are Gaussian random vectors, but leaves open the question of whether $R(X,Y)$ can be related to some
more familiar distance from Gaussianity. Even more recently, Courtade, Fathi and Pananjady \cite{CFP16}
showed that if $X$ and $Y$ are uniformly log-concave (in the sense that the densities of both are of the form
$e^{-V}$ with the Hessian of $V$ bounded from below by a positive multiple of the identity matrix), then the deficit
in the EPI is controlled in terms of the quadratic Wasserstein distances between the distributions of $X$ and $Y$ and
Gaussianity.

There are also strengthenings of the EPI when one of the summands is Gaussian.
Set $X^{(t)}=X+\sqrt{t}Z$,
with $Z$ a standard Gaussian random variable independent of $X$.
Costa \cite{Cos85b} showed that for any $t\in [0,1]$, 
\be\label{costa-ineq}
N(X^{(t)})\geq (1-t)N(X)+tN(X+Z) .
\ee
This may be rewritten as
$N(X^{(t)})- N(X) \geq t [N(X+Z)-N(X)] 
=N(\sqrt{t}X+\sqrt{t}Z)-N(\sqrt{t}X)$.
Setting $\beta=\sqrt{t}$, we have for any $\beta\in [0,1]$ that
$
N(X+\beta Z) - N(X) \geq N(\beta X+\beta Z)-N(\beta X) ,
$
substituting $X$ by $\beta X$, we get
\be\label{inq:CostaMono1}
N(X+Z) - N(X) \geq N(\beta X+ Z)-N(\beta X) .
\ee
for any $\beta\in [0,1]$. Therefore, for any $\beta$, $\beta'\in [0,1]$ with $\beta>\beta'$, substitute $X$ by $\beta X$ and $\beta$ by $\beta'/\beta$ in \eqref{inq:CostaMono1}, we have
\ben
N(\beta X+Z) - N(\beta X) \geq N(\beta' X+ Z)-N(\beta' X) .
\een
In other words, Costa's result states that if $A(\beta)=N(\beta X+ Z)-N(\beta X)$,
then $A(\beta)$ is a monotonically increasing function for $\beta\in [0,1]$.
To see that this is a refinement of the EPI in the special case when one summand is Gaussian,
note that the EPI in this case is the statement that $A(1)\geq A(0)$. An alternative proof
of Costa's inequality was given by Villani~\cite{Vil00}; for a generalization, see \cite{PP09}.

Very recently, a powerful extension of Costa's inequality was developed by Courtade~\cite{Cou16}, 
applying to a system in which $X, X+Z, V$ form a Markov chain (i.e., $X$ and $V$ are conditionally independent
given $X+Z$) and $Z$ is a Gaussian random vector independent of $X$. Courtade's result
specializes in the case where $V=X+Z+Y$ to the following:
If $X, Y, Z$ be independent random vectors in $\R^d$ with $Z$ being Gaussian, then
\be\label{eq:courtade}
N(X+Z) N(Y+Z) \geq N(X) N(Y) + N(X+Y+Z) N(Z) .
\ee
Applying the inequality \eqref{eq:courtade} 
to $X$, $\sqrt{1-t}Z'$ and $\sqrt{t}Z$ where $Z'$ is the independent copy of the standard normal distribution $Z$, we have
\ben
N(X^{(t)})N(\sqrt{1-t}Z'+\sqrt{t}Z)\ge N(X)N(\sqrt{1-t}Z')+N(X+\sqrt{1-t}Z'+\sqrt{t}Z)N(\sqrt{t}Z).
\een
By the fact that $\sqrt{1-t}Z'+\sqrt{t}Z$ has the same distribution as $Z$, and by the fact that $N(Z)=1$, we have
$N(X^{(t)})\ge (1-t)N(X)+tN(X+Z)$,
which is Costa's inequality \eqref{costa-ineq}.

Motivated by the desire to prove entropic central limit theorems for statistical physics models,
some extensions of the EPI to dependent summands have also been considered (see, e.g., \cite{CS91, Tak96, Tak98, Joh04:1, Joh06}), 
although the assumptions tend to be quite restrictive for such results. 

Finally there is an extension of the EPI that applies not just to sums but to more general linear transformations
applied to independent random variables. The main result of Zamir and Feder \cite{ZF93} asserts that if $X_1, \ldots, X_n$
are independent real-valued random variables, $Z_1, \ldots, Z_n$ are independent Gaussian random variables
satisfying $h(Z_i)=h(X_i)$, and $A$ is any matrix,
then $h(AX) \geq h(AZ)$ where $AX$ represents the left-multiplication of the vector $X$ by the matrix $A$.
As explained in \cite{ZF93}, for this result to be nontrivial, 
the $m\times n$ matrix $A$ must have $m< n$ and be of full rank.
To see this, notice that if $m>n$ or if $A$ is not of full rank,
the vector $AX$ does not have full support on $\RL^m$ and $h(AX)=h(AZ)=-\infty$,
while if $m=n$ and $A$ is invertible, $h(AX)=h(AZ)$ holds with equality because of the
conditions determining $Z$ and the way entropy behaves under linear transformations.

\subsection{R\'enyi Entropy Power inequalities}
\label{sec:epi-r}

\subsubsection{First R\'enyi interpolation of the EPI and BMI}

Unified proofs can be given of the EPI and the BMI in different ways, each of which
may be thought of as providing extensions of the EPI to R\'enyi entropy.

The first unified approach is via Young's inequality. 
Denote by $L^p$ the Banach space $L^p(\RL^d,dx)$ of measurable functions
defined on $\RL^d$ whose $p$-th power is integrable with respect to Lebesgue measure $dx$.
In 1912, Young \cite{You12} introduced the fundamental inequality
\be\label{orig-young}
\|f\star g\|_r \leq \|f\|_p \|g\|_q \,,\quad \frac{1}{p}+\frac{1}{q}=\frac{1}{r} +1, \quad 1< p,q,r< +\infty ,
\ee
for functions $f\in L^p$ and $g\in L^q$,
which implies that if two functions are in (possibly different) $L^p$-spaces, then
their convolution is contained in a third $L^p$-space.
In 1972, Leindler \cite{Lei72a} showed the so-called reverse Young inequality,
referring to the fact that the inequality \eqref{orig-young} is reversed when
$0<p,q,r<1$.
The best constant that can be put on the right side of
\eqref{orig-young} or its reverse was found by Beckner \cite{Bec75}: 
the best constant  is $(C_p C_q/ C_{r})^d$, where 
\be\label{eq:Cp}
C_{p}^{2}=\frac{p^{\frac{1}{p}}}{|p'|^{\frac{1}{p'}}} ,
\ee
and for any $p\in (0,\infty]$, $p'$ is defined by
\be\label{eq:dual}
\frac{1}{p}+\frac{1}{p'}=1 .
\ee
Note that $p'$ is positive for $p\in (1,\infty)$, and negative for $p\in (0,1)$. 
Alternative proofs of both Young's inequality and the reverse Young inequality with this sharp constant
were given by Brascamp and Lieb \cite{BL76b}, Barthe \cite{Bar98a}, and Cordero-Erausquin
and Ledoux \cite{CL10}.

We state the sharp Young and reverse Young inequalities  now for later reference.

\begin{thm}\cite{Bec75}
Suppose $r\in(0,1)$ and $p_{i}\in (0,1)$ satisfy
\be\label{eq:p-cond}
\sum_{i=1}^n  \frac{1}{p_i} = n- \frac{1}{r'} .
\ee
Then, for any functions $f_j\in L^{p_j}$ ($j=1, \ldots, n$),
\be\label{eq:reverse}
\bigg\| \star_{j\in [n]} f_{j} \bigg\|_{r} \geq
\frac{1}{C_{r}^d} \prod_{j\in [n]} \big[ C_{p_{j}}^d
\| f_{j}  \|_{p_{j}} \big] .
\ee
The inequality is reversed if $r\in (1,\infty)$ and $p_{i}\in (1,\infty)$.
\end{thm}

Dembo, Cover and Thomas \cite{DCT91} interpret the Young and reverse Young inequalities with sharp constant
as EPI's for the R\'enyi entropy.
If $X_i$ are random vectors in $\RL^d$ with densities $f_i$ respectively,
taking the logarithm of \eqref{eq:reverse}
and rewriting the definition 
of the R\'enyi entropy power as
$N_p(X)=\|f\|_p^{-2p'/d}$, we have
\be\label{eq:log-conj}\begin{split}
\frac{d}{2r'}\log N_r\bigg(\sum_{i\in[n]} X_i\bigg)
\leq d\log C_{r} &- {d}\sum_{i\in[n]} \log C_{p_{i}}
+  \sum_{i\in[n]} \frac{d}{2p_{i}'} \log N_{p_{i}}( X_i) .
\end{split}\ee
Introduce two discrete probability measures $\lam$ and $\kappa$ on $[n]$, with
probabilities proportional to $1/p_{i}'$ and $1/p_{i}$ respectively. 
Setting 
$L_r=rn-r+1= r(n-1/r')$, the condition \eqref{eq:p-cond},
allows us to write explicitly 
\ben\begin{split}
\kappa_{i}  
&= \bigg(\frac{r}{L_r}\bigg) \frac{1}{p_{i}}  ,\\
\lambda_{i}&= \frac{r'}{p_{i}'}  ,
\end{split}\een
for each $i\in [n]$, also using $1/p_{i}+1/p_{i}' =1$ for the latter.
Then \eqref{eq:log-conj} reduces to
\ben
h_r(Y_{[n]})
\geq \frac{dr'}{2} \log C_{r}^2 - \frac{dr'}{2}\sum_{i\in[n]}  \log C_{p_{i}}^2 + \sum_{i\in[n]}  \lambda_{i} h_{p_{i}}(X_i) .
\een
Now, some straightforward calculations show that if we take the limit
as $p_i, r\ra 0$ from above, we get the BMI, while if we take the limit
as $p_i, r\ra 1$, we get the EPI (this was originally observed by Lieb \cite{Lie78}).

\subsubsection{Second R\'enyi interpolation of the EPI and BMI}

Wang and Madiman \cite{WM14} found a rearrangement-based refinement of the EPI that also applies to R\'enyi entropies.
For a Borel set $A$, define its spherically decreasing symmetric rearrangement $A^*$ by
\begin{eqnarray*}
A^*:=B(0,r),
\end{eqnarray*}
where $B(0,r)$ stands for the open ball with radius $r$ centered at the origin and $r$ is determined by the condition that $B(0,r)$ has volume $|A|$. Here we use the convention that if $|A|=0$ then $A^*=\emptyset$ and that if $|A|=\infty$ then $A^*=\mathbb{R}^d$. Now for a measurable non-negative function $f$, define its spherically decreasing symmetric rearrangement $f^*$ by
\begin{eqnarray*}
f^*(y):=\int_{0}^\infty\mathbbm{1}_{\{y\in B_t^*\}}dt,
\end{eqnarray*}
where $B_t:=\{x:~f(x)>t\}$. 
It is a classical fact (see, e.g., \cite{Bur09:tut}) that rearrangement preserves $L^p$-norms, i.e.,
$\|f^*\|_p=\|f\|_p$. In particular, if $f$ is a probability density function, so is $f^*$. 
If $X\sim f$, denote by $X^*$ a random variable  with density $f^*$; then the rearrangement-invariance of 
$L^p$-norms immediately implies that $h_p(X^*)= h_p(X)$ for each $p\in [0,\infty]$ (for $p=1$,
this is not done directly but via a limiting argument).

\begin{thm}\cite{WM14}\label{thm:rearr}
Let $X_1, \ldots, X_n$ be independent $\R^d$-valued random vectors. Then
\begin{eqnarray}\label{inq:RearrgmtRenyi}
h_p(X_1+\ldots + X_n)\ge h_p(X_1^*+\ldots+ X_n^*)
\end{eqnarray}
for any $p\in [0,\infty]$, provided the entropies exist.
\end{thm}

In particular, 
\begin{eqnarray}\label{inq:EPIRea}
N(X+Y)\ge N(X^*+Y^*),
\end{eqnarray}
where $X$ and $Y$ are independent random vectors with density functions $f$ and $g$ respectively and $X^*$ and $Y^*$ are independent random vectors with density function $f^*$ and $g^*$ respectively. Thanks to \eqref{inq:EPIRea}, we have effectively inserted an intermediate term in between the two sides of the formulation \eqref{inq:eqv2} of the EPI:
\begin{eqnarray*}
N(X+Y)\ge N(X^*+Y^*)\ge N(X^G+Y^G),
\end{eqnarray*}
where the second inequality is by the fact that $h(X^G)=h(X^*)=h(X)$, combined with
the third equivalent form of the EPI in Theorem~\ref{thm:epi}.
In fact, it is also shown in \cite{WM14} that the EPI itself can be deduced from \eqref{inq:EPIRea}.

\subsubsection{A conjectured R\'enyi EPI}

Let us note that neither of the above unifications of BMI and EPI via R\'enyi entropy directly gives
a sharp bound on $N_p(X+Y)$ in terms of $N_p(X)$ and $N_p(Y)$. The former approach
relates R\'enyi entropy powers of different indices, while the latter refines the third formulation
in Theorem 2.1 (but not the first, because the equivalence that held for Shannon-Boltzmann entropy
does not work in the R\'enyi case). The question of finding a sharp direct relationship between 
$N_p(X+Y)$ with $N_p(X)$ and $N_p(Y)$ remains open, with some non-sharp results for the $p>1$ case obtained 
by Bobkov and Chistyakov \cite{BC15:1}, whose argument and results were recently tightened by Ram and Sason \cite{RS16}.

\begin{thm}\label{thm:rs}\cite{RS16}
For $p \in (1, \infty)$ and independent random vectors $X_i$ with densities in $\mathbb{R}^d$, 
\[
N_p(X_1 + \cdots +X_n) \geq c_p^{(n)} \sum_{i=1}^n N_p(X_i) ,
\]
where $p' = p/(p-1)$ and
\[
	c_p^{(n)} = p^{\frac{1}{p-1}} \left( 1- \frac 1 {np'} \right)^{np'-1} \geq \frac{1}{e}.
\]
\end{thm}

We now discuss a conjecture of Wang and Madiman \cite{WM14} about extremal distributions 
for R\'enyi EPI's of this sort. Consider the one-parameter family
of distributions, indexed by a parameter $-\infty<\beta\leq \frac{2}{d+2}$, of the following form:
$g_0$ is the standard Gaussian density in $\RL^d$, and for $\beta\neq 0$,
\ben
g_\beta(x)= A_{\beta} \bigg(1-\frac{\beta}{2}\|x\|^2 \bigg)_{+}^{\frac{1}{\beta}-\frac{d}{2}-1} ,
\een
where $A_{\beta}$ is a normalizing constant
(which can be written explicitly in terms of gamma functions).
We call $g_\beta$ the {\it standard generalized Gaussian} of order $\beta$;
any affine function of a standard generalized Gaussian yields a ``generalized Gaussian''.
The densities $g_\beta$ (apart from the obviously special value $\beta=0$)
are easily classified into two distinct ranges where they behave differently. {\it First}, for $\beta<0$,
the density is proportional to a negative power of $(1+b\|x\|^2)$ for a positive constant $b$,
and therefore correspond to measures with full support on $\RL^d$ that are heavy-tailed.
For $\beta>0$, note that $(1-b\|x\|^2)_+$ with positive $b$ is non-zero only for $\|x\|<b^{-\half}$, and is concave in this region.
Thus any density in the {\it second} class, corresponding to 
$0<\beta\leq\frac{2}{d+2}$,
is a positive power of $(1-b\|x\|^2)_+$, and is thus a concave function supported on a centered Euclidean ball
of finite radius.
It is pertinent to note that although the first class includes many
distributions from what one might call the ``Cauchy family'', it excludes the standard Cauchy distribution;
indeed, not only do all the generalized Gaussians defined above have finite variance, but in fact the
form has been chosen so that, for $Z\sim g_\beta$,
\ben
\E [\|Z\|^2 ]=d
\een
for any $\beta$.
The generalized Gaussians have been called by different names in the literature,
including Barenblatt profiles, or  the Student-$r$ distributions ($\beta<0$)
and Student-$t$  distributions ($0<\beta\leq \frac{2}{d+2}$).

For $p>\frac{d}{d+2}$, define $\beta_p$ by
\ben
\frac{1}{\beta_p}= \frac{1}{p-1}+\frac{d+2}{2} ,
\een
and write $Z^{(p)}$ for a random vector drawn from $g_{\beta_p}$. Note that $\beta_p$ ranges from $-\infty$ to $\frac{2}{d+2}$ as $p$ ranges from
$\frac{d}{d+2}$ to $\infty$.  The generalized Gaussians $Z^{(p)}$ arise naturally
as the maximizers of the R\'enyi entropy power of order $p$
under a variance constraint, as independently observed by Costa, Hero and Vignat \cite{CHV03} 
and Lutwak, Yang and  Zhang \cite{LYZ07a}.
They play the starring role in the conjecture of Wang and Madiman \cite{WM14}.

\begin{conj}\label{conj:renyi-epi}\cite{WM14}
Let $X_1, \ldots, X_n$ be independent random vectors taking values in $\RL^d$,
and $p> \frac{d}{d+2}$. Suppose $Z_i$ are independent random vectors, each a scaled version of $Z^{(p)}$.
such that $h_p(X_i)=h_p(Z_i)$. Then
\ben
N_p(X_1+\ldots+X_n) \geq N_p(Z_1+\ldots+Z_n) .
\een
\end{conj}

Until very recently, this conjecture was only known to be true in the case where $p=1$ (when it is the classical EPI)
and the case where $p=\infty$ and $d=1$ (which is due to Rogozin \cite{Rog87:1} and discussed in Section~\ref{sec:epi-inf}).
In \cite{MMX16:2}, we have very recently been able to prove Conjecture~\ref{conj:renyi-epi} for $p=\infty$ and any finite dimension $d$,
generalizing Rogozin's inequality. All other cases remain open.

\subsubsection{Other work on R\'enyi entropy power inequalities}

Johnson and Vignat \cite{JV07} also demonstrated what they
call an ``entropy power inequality for R\'enyi entropy'', for any order $p\geq 1$. However, their
inequality does not pertain to the usual convolution, but a new and somewhat complicated
convolution operation (depending on $p$). This new operation reduces to the usual convolution for $p=1$,
and has the nice property that the convolution of affine transforms of independent copies
of $Z^{(p)}$ is an  affine transform of $Z^{(p)}$ (which fails for the usual convolution when $p> 1$).

As discussed earlier, Costa \cite{Cos85b} proved a strengthening of the classical EPI
when one of the summands is Gaussian.
Savar\'e and Toscani  \cite{ST14} recently proposed a generalization of Costa's
result to R\'enyi entropy power, but the notion of concavity they use based on
solutions of a nonlinear heat equation does not have obvious probabilistic meaning. 
Curiously, it turns out that the definition of R\'enyi entropy power appropriate
for the framework of \cite{ST14} has a different constant in the exponent 
($\frac{2}{d}+p-1$ as opposed to $\frac{2}{d}$).
Motivated by \cite{ST14}, Bobkov and Marsiglietti \cite{BM16} very recently
proved R\'enyi entropy power inequalities with non-standard exponents. Their main result
may be stated as follows.

\begin{thm}\label{thm:bm-renyi}\cite{BM16}
For $p \in (1, \infty)$ and independent random vectors $X_i$ with densities in $\mathbb{R}^d$, 
\[
\tilde{N}_p(X_1 + \cdots +X_n) \geq  \sum_{i=1}^n \tilde{N}_p(X_i) ,
\]
where 
\[
	\tilde{N}_p(X) = e^{\frac{p+1}{d} h_p(X)}.
\]
\end{thm}

It would be interesting to know if Theorem~\ref{thm:bm-renyi} is true for $p\in[0,1)$
(and hence all $p\geq 0$), since this would be a particularly nice interpolation between the BMI and EPI.

It is natural to look for R\'enyi entropy analogues of the refinements and generalizations 
of the EPI discussed in Section~\ref{sss:fancy-epi}. While little has been done in this direction
for general R\'enyi entropies (apart from the afore-mentioned work of \cite{ST14}), the 
case of the R\'enyi entropy of order 0 (i.e., inequalities for volumes of sets)-- which is, 
of course, of special interest-- has attracted some study. For example,
Zamir and Feder \cite{ZF98} demonstrated  a nontrivial version of the BMI for sums of the form
$v_1 A_1 + \ldots v_k A_k$, where $A_i$ are unit length subsets of $\RL$ and $v_i$ are vectors
in $\RL^d$, showing that the volume of the Minkowski sum is minimized when each $A_i$
is an interval (i.e., the sum is a zonotope). This result was motivated by analogy with the ``matrix version''
of the EPI discussed earlier.

Indeed, the strong parallels between the BMI and the EPI might lead to the belief that every volume inequality
for Minkowski sums has an analogue for entropy of convolutions, and vice versa. 
However, this turns out not to be the case. It was shown by Fradelizi and Marsiglietti \cite{FM14}
that the analogue of Costa's result  \eqref{costa-ineq} on concavity of entropy power,
namely the assertion that $t\mapsto |A+tB_2^d|^\dth$ is concave for positive $t$
and any given Borel set $A$, fails to hold\footnote{They also showed some partial positive results--
concavity holds in dimension 2 for connected sets, and in general dimension on a subinterval $[t_0,\infty)$
under some regularity conditions.} even in dimension 2. Another conjecture in this spirit that
was made independently by V.~Milman (as a generalization of Bergstrom's determinant inequality)
and by Dembo, Cover and Thomas \cite{DCT91} (as an analogue of Stam's Fisher information inequality,
which is closely related to the EPI) was disproved by Fradelizi, Giannopoulos and Meyer \cite{FGM03}.
In \cite{BMW11}, it was conjectured that analogues of fractional EPI's such as \eqref{inq:RFMEPI} hold for volumes,
and it was observed that this is indeed the case for convex sets. If this conjecture were true for general compact sets,
it would imply that for any compact set, the volumes of the Minkowski self-averages (obtained by taking the Minkowski sum of $k$ copies
of the set, and scaling by $1/k$) are monotonically increasing\footnote{The significance of this arises from the
fact that the Minkowski self-averages of any compact set converge in Hausdorff distance to the convex hull of the set,
and furthermore, one also has convergence of the volumes if the original compact set had nonempty interior.
Various versions of this fact were proved independently by Emerson and Greenleaf \cite{EG69}, and by Shapley, Folkmann and Starr \cite{Sta69};
a survey of such results including detailed historical remarks can be found in \cite{FMMZ16}.} in $k$.
However, \cite{FMMZ16:cras} showed that this conjecture does not hold\footnote{On the other hand, partial positive results quantifying the
convexifying effect of Minkowski summation were obtained in \cite{FMMZ16:cras, FMMZ16}.} in general-- 
in fact, they showed that there exist many compact sets $A$ in $\R^d$
for any $d\geq 12$ such that $|A+A+A| < (\frac{3}{2})^d |A+A|$. Finally while volumes of Minkowski sums
of convex sets in $\RL^d$ are supermodular (as shown in \cite{FMMZ16}), entropy powers of convolutions of log-concave 
densities fail to be supermodular even in dimension 1 (as shown in \cite{MG16}).
Thus the parallels between volume inequalities and entropy inequalities are not exact.

Another direction that has seen considerable exploration in recent years is stability of the BMI.
This direction began with stability estimates for the BMI in the case where the two summands are convex sets \cite{Dis73, Gro88, FMP09, FMP10, Seg12}\footnote{There
is also a stream of work on stability estimates for other geometric inequalities related to the BMI, such as the isoperimetric inequality,
but this would take us far afield.},
asserting that near-equality in the BMI implies that the summands are nearly homothetic.
For general Borel sets, qualitative stability (i.e., that closeness to equality entails closeness to extremizers) was shown by Christ \cite{Chr12:1, Chr12:2}, 
with the first quantitative estimates recently developed by Figalli and Jerison \cite{FJ15}.
Qualitative stability for the more general Young's inequality has also been recently considered \cite{Chr11},
but quantitative estimates are unknown to the extent of our knowledge.

\subsection{An EPI for R\'enyi entropy of order $\infty$}
\label{sec:epi-inf}

In discussing R\'enyi entropy power inequalities, it is of particular interest to consider
the case of $p=\infty$, because of close connections with the literature in probability theory
on small ball estimates and the so-called L\'evy concentration functions \cite{NV13, EGZ15},
which in turn have applications to a number of areas including stochastic process theory \cite{LS01}
and random matrix theory \cite{RV09, TV09:1, RV10:icm}.

Observe that by Theorem \ref{thm:CdtEntp} we trivially have 
\be\label{eq:2sum-inf}
N_\infty(X+Y) \geq \max \{N_\infty(X), N_\infty(Y) \} \geq \frac 1 2 (N_\infty(X)+N_\infty(Y)).
\ee
In fact, the constant $\frac 1 2$ here is sharp, as uniform distributions on any symmetric convex set $K$ (i.e., $K$ is convex,
and $x\in K$ if and only if $-x\in K$) 
of volume $1$  are extremal: if $X$ and $X'$ are independently distributed according to $f=\mathbbm{1}_K$,
then denoting the density of $X-X'$ by $u$, we have
\ben
\|u\|_{\infty}=u(0)=\int f^2(x) dx = 1 = \|f\|_{\infty} ,
\een
so that $N_{\infty}(X+X')=N_{\infty}(X-X')= N_{\infty}(X)=\half [N_{\infty}(X)+N_{\infty}(X')]$.

What is more, it is observed in \cite{BC15:1} that when each $X_i$ is real-valued, $1/2$ is the optimal constant for any number of summations. 

\begin{thm}\cite{BC15:1}\label{1dinfinityEPI}
For independent, real-valued random variables $X_1, \dots, X_n$,
\ben
	N_\infty \left(\sum_{i=1}^n X_i \right) \geq \frac 1 2 N_\infty (X_i) .
\een
\end{thm}

The constant $1/2$ clearly cannot be improved upon (one can take $X_3, \ldots, X_n$ to be deterministic and the result follows from the $n=2$ case).  
That one should have this sort of scaling in $n$ for the lower bound (namely, linear in $n$ when the summands are identically distributed with bounded densities)
is not so obvious from the trivial maximum bound above. 
The proof of Theorem~\ref{1dinfinityEPI} draws on two theorems, the first due to Rogozin \cite{Rog87:1}, which reduces the general case to the cube, 
and the second a geometric result on cube slicing due to K.~Ball \cite{Bal86}.

\begin{thm}\cite{Rog87:1}\label{rogsquare}
Let $X_1, \dots, X_n$ be independent $\mathbb{R}$-valued random variables with bounded densities.  Then 
	\begin{align} 
		N_\infty(X_1 + \cdots + X_n) \geq N_\infty(Y_1 + \cdots + Y_n),
	\end{align}
where $Y_1, \dots, Y_n$ are a collection of independent random variables, with $Y_i$ chosen to be uniformly distributed on a symmetric interval such that $N_\infty(Y_i) = N_\infty(X_i)$.
\end{thm}

\begin{thm}\cite{Bal86} \label{ballsquare}
	Every section of the unit cube $[-\frac 1 2, \frac 1 2]^d$ denoted $Q_d$ by an $(d-1)$-dimensional subspace has volume bounded above by $\sqrt{2}$.  This upper bound is attained iff the subspace contains a $(d-2)$-dimensional face of $Q_d $.
\end{thm}

\begin{proof}[Proof of Theorem \ref{1dinfinityEPI}]
For $X_i$ independent and $\RL$-valued, with $Y_i$ chosen as in Theorem~\ref{rogsquare},
\begin{align*}
N_\infty(X_1 + \cdots + X_n) \geq N_\infty(Y_1 + \cdots + Y_n).
\end{align*}
Applying a sort of change of variables, and utilizing the degree $2$ homogeneity of entropy powers, one can write
\[
N_\infty(Y_1 + \cdots + Y_n) = \left(\sum_{i=1}^n N_\infty(Y_i) \right) N_\infty(\theta_1 U_1 + \cdots + \theta_n U_n),
\] 
where the $U_i$ are independent uniform on $[-\frac 1 2, \frac 1 2]$ and $\theta$ is a unit vector (to be explicit, take $\theta_i = \sqrt{N_\infty(Y_i)/\sum_j N_\infty(Y_j)}$ and the above can be verified). Then utilizing the symmetry of $\theta_1 U_1 + \cdots + \theta_n U_n$ and the BMI, we see that the maximum of its density must occur at $0$, yielding
\ben
N_\infty(\theta_1 U_1 + \cdots + \theta_n U_n)
=\left|Q_d \cap \theta^\perp \right|^{-2}_{d-1} \ge \frac{1}{2}.
\een
The result follows.
\end{proof}

Theorem \ref{1dinfinityEPI} admits two natural generalizations.  The first, also handled in \cite{BC15:1} (and later recovered 
in \cite{RS16} by taking the limit as $p\ra\infty$ in Theorem~\ref{thm:rs}), is the following.  

\begin{thm} \cite{BC15:1} \label{BC15thm}
For independent random vectors $X_1, \dots, X_n$ in $\mathbb{R}^d$.  
\begin{align} \label{BCinfty}
N_\infty(X_1 + \cdots + X_n) 
	&\geq 
		\bigg(1-\frac{1}{n}\bigg)^{n-1} [N_\infty(X_1)+ \cdots + N_\infty(X_n)]
		\\
	&\geq
		\frac 1 e [N_\infty(X_1)+ \cdots + N_\infty(X_n)].
\end{align}
\end{thm}  

A second direction was pursued by Livshyts, Paouris and Pivovarov \cite{LPP15} in which the authors derive 
sharp bounds for the maxima of densities obtained as the projections of product measures. Specifically, 
\cite[Theorem 1.1]{LPP15} shows that  given probability density functions $f_i$  on $\mathbb{R}$ with 
$\|f_i\|_\infty \leq 1$, with joint product density $f$ defined by $f(x_1,\ldots, x_n) = \prod_{i =1}^n f_i(x_i)$, then 
\begin{align} \label{LPP}
 \| \pi_E(f) \|_\infty \leq \min \left( \left(\frac{n}{n-k}\right)^{(n-k)/2}, 2^{k/2} \right) ,
\end{align}
where $\pi_E(f)$ denotes the pushforward of the probability measure induced by $f$ under orthogonal projection to a $k$-dimensional subspace $E$,
i.e., $\pi_E(f)(x) = \int_{x + E^\perp} f(y) dy $. In addition, cubes are shown to be extremizers of the above inequality.  
In the language of information theory, this can be rewritten as follows.

\begin{thm}\cite{LPP15} \label{LPP15thm}
Let $X = (X_1, \ldots, X_n)$ where $X_i$ are independent $\R$-valued random variables, and $N_\infty(X_i) \geq 1$. 
Then
\begin{align} \label{inftyEPI}
N_\infty(P_E X) \geq \max \left\{\frac 1 2 , \left( 1-\frac{k}{n} \right)^{\frac{n}{k}-1}  \right\} ,
\end{align}
where $P_E$ denotes the orthogonal projection to a $k$-dimensional subspace $E$, 
and equality can be achieved for $X_i$ uniform on intervals.
\end{thm}

In the $k=1$ case, this implies Theorem \ref{1dinfinityEPI} by applying the inequality \eqref{inftyEPI} to 
\ben
Y_i = X_i / \sqrt{N_\infty(X_i)} ,
\een 
and taking $E$ to be the space spanned by the unit vector 
$\theta_i = \sqrt{ N_\infty(X_i)/\sum_j N_\infty(X_j)}$.  
The $Y_i$ defined satisfy the hypothesis so we have $N_\infty(P_E Y) \geq 1/2$, but 
\begin{align*}
	N_\infty(P_E Y) 
		&=
			N_\infty( \langle \theta , Y\rangle)
			\\
		&=
			N_\infty\left( \frac{ X_1 + \cdots + X_n}{ \sqrt{\sum_{j=1}^n N_\infty(X_j)}} \right)	
			\\
		&=
			\frac{ N_\infty( X_1 + \cdots + X_n)}{ \sum_{j=1}^n N_\infty(X_j)},
\end{align*} 
and the implication follows.

Conversely, for the one-dimensional subspace $E$ spanned by the unit vector $\theta$, and $X_i$ satisfying $N_\infty(X_i)\geq 1$, 
if one applies Theorem~\ref{1dinfinityEPI} to $Y_i = \theta_i X_i$, we recover the one-dimensional case of the projection theorem as
\begin{align*}
	N_\infty(P_E X)
		&=
			N_\infty(Y_1 + \cdots + Y_n) 
			\\
		&\geq
			\frac 1 2 ( N_\infty(Y_1) + \cdots + N_\infty(Y_n) )
			\\
		&=
			\frac 1 2 ( \theta_1^2 N_\infty(X_1) + \cdots + \theta_n^2 N_\infty(X_n) )
			\\
		&\geq
			\frac 1 2.
\end{align*} 

Thus Theorem~\ref{LPP15thm} can be seen as a $k$-dimensional generalization of the $\infty$-EPI for real random variables. 

In recent work \cite{MMX16:2}, we have obtained a generalization of Rogozin's inequality that allows us to prove
multidimensional versions of both Theorems~\ref{BC15thm} and \ref{LPP15thm}. Indeed, our extension of 
Rogozin's inequality reduces both the latter theorems to geometric inequalities about Cartesian products of Euclidean balls, 
allowing us to obtain sharp constants in Theorem~\ref{1dinfinityEPI} for any fixed dimension
as well as to generalize Theorem~\ref{LPP15thm} to the case where each $X_i$ is a random vector.

\section{Reverse Entropy Power Inequalities}
\label{sec:repi}

\subsection{$\kappa$-concave measures and functions}
\label{sec:cvx}

$\kappa$-concave measures are measures that satisfy a generalized Brunn-Minkowski inequality, 
and were studied systematically by Borell \cite{Bor74, Bor75a}.  

As a prerequisite, we define the $\kappa$-mean of two numbers, for $a,b \in (0,\infty)$, $t \in (0,1)$ and $\kappa \in (-\infty,0)\cup (0,\infty)$ define 
\begin{align}
M_\kappa^t(a,b) = \left((1-t) a^{\kappa} + t b^\kappa \right)^{\frac 1 \kappa}.
\end{align}
For $\kappa \in \{-\infty, 0, \infty \}$ define $M_\kappa^t(a,b) = \lim_{\kappa' \to \kappa} M_{\kappa'}^t(a,b)$ corresponding to 
\[
\{ \min(a,b), a^{1-t}b^t, \max(a,b)\}
\] respectively.  $M_\kappa$ can be extended to $a,b \in [0,\infty)$ via direct evaluation when $\kappa \geq 0$ and again by limits when $\kappa < 0$ so that $M_\kappa(a,b) = 0$ whenever $ab=0$.  

\begin{defn}
Fix $\kappa \in [-\infty, \frac 1 d]$. We say that a probability measure $\mu$ on $\mathbb{R}^d$ is $\kappa$-concave if 
the support of $\mu$ has non-empty 
interior\footnote{We only assume this for simplicity of exposition-- a more  general theory not requiring absolute continuity of the measure $\mu$
with respect to Lebesgue measure on $\RL^d$ is available in Borell's papers. Note that while the support of $\mu$ having nonempty interior in
general is a weaker condition than absolute continuity, the two conditions turn out to coincide in the presence of a $\kappa$-concavity assumption.},
and
\begin{align*}
\mu((1-t)A+tB) \geq M_\kappa^t(\mu(A),\mu(B))
\end{align*}
for any Borel sets $A, B$, and any $t \in [0,1]$.

We say that $\mu$ is a convex measure if it is $\kappa$-concave for some $\kappa\in [-\infty, \frac 1 d]$.  

When the law of a random vector $X$ is a $\kappa$-concave measure, 
we will refer to $X$  as a $\kappa$-concave random vector. 
\end{defn}

Thus, the $\kappa$-concave measures are those that distribute volume in such a way that the vector space average of two sets is larger than the $\kappa$-mean of their respective volumes.  Let us state some preliminaries. First notice that by Jensen's inequality $\mu$ being $\kappa$-concave implies $\mu$ is $\kappa'$-concave for $\kappa' \leq \kappa$. The support of a $\kappa$-concave measure is necessarily convex, and since we assumed that the support has nonempty interior, the dimension of the smallest affine subspace of $\RL^d$ containing the support of $\mu$ is automatically $d$.   

It is a nontrivial fact that concavity properties of a measure can equivalently be described pointwise in terms of its density.

\begin{thm}[\cite{Bor74}]\label{thm:k-kd}
A measure $\mu$ on $\RL^d$ is $\kappa$-concave if and only if it has a density (with respect to the Lebesgue measure on its support)
that is a $s_{\kappa,d}$-concave function, in the sense that
\[
f((1-t)x+ty) \geq M_{s_{\kappa,d}}^t(f(x),f(y))
\]
whenever $f(x)f(y)>0$ and $t\in (0,1)$, and where
\[
s_{\kappa,d} := \frac \kappa {1-\kappa d} .
\]
\end{thm}

\noindent
{\bf Examples:}
\begin{enumerate}
\item If $X$ is the uniform distribution on a convex body $K$, it has an $\infty$-concave density function 
$f = |K|^{-1} \mathbbm{1}_K$ and thus the probability measure is $1/d$-concave.  
Let us note that by our requirement that $\mu$ is ``full-dimensional" (i.e., has support with nonempty interior), the only 
$1/d$-concave probability measures on $\mathbb{R}^d$ are of this type.

\item A measure that is $0$-concave is also called a {\it log-concave measure}. Since $s_{0,d}=0$ for any positive integer $d$, Theorem~\ref{thm:k-kd}
implies that an absolutely continuous measure $\mu$ is log-concave if and only if its density is a log-concave function (as defined in 
Definition~\ref{defn:lc-fn}).
In other words, $X$ has a log-concave distribution if and only if its density function can be expressed on its support as $e^{-V(x)}$ for $V$ convex.  
When $V(x) =\half |x|^2- \frac{d}{2}\log (2\pi)$, one has the standard Gaussian distribution; 
when $V(x) = x$ for $x\geq 0$ and $V(x) = \infty$ for $x<0$, one has the standard exponential distribution; and so on.

\item If $X$ is log-normal distribution with density function
\ben
f(x):=\frac{1}{x\sigma\sqrt{2\pi}}e^{-\frac{(\ln x-\mu)^2}{2\sigma^2}}
\een
Then the density function of $X$ is $-\frac{\sigma}{4}$-concave, and for $\sigma<4$,  the probability measure is $\frac{-\sigma}{4-\sigma}$-concave.

\item If $X$ is a Beta distribution with density function
\ben
\frac{x^\alpha(1-x)^\beta}{B(\alpha,\beta)}
\een
with shape parameters $\alpha\ge 1$ and $\beta\ge 1$, then the density function of $X$ is $\min(\frac{1}{\alpha-1},\frac{1}{\beta-1})$-concave, and the probability measure is $\frac{1}{\max(\alpha,\beta)}$-concave.
\item If $X$ is a $d$-dimensional Student's $t$-distribution with density function
\ben
f(x):=\frac{\Gamma(\frac{\nu+d}{2})}{\nu^{\frac{d}{2}}\pi^{\frac{d}{2}}\Gamma(\frac{\nu}{2})}\left(1+\frac{|x|^2}{\nu}\right)^{-\frac{\nu+d}{2}}
\een
with $\nu>0$, then the density function of $X$ is $-\frac{1}{\nu+d}$-concave, and the probability measure is $-\frac{1}{\nu}$-concave.
\item If $X$ is a $d$-dimensional Pareto distribution of the first kind with density function
\ben
f(x):=a(a+1)\cdots(a+d-1)\left(\prod_{i=1}^d\theta_i\right)^{-1}
\left(\sum_{i=1}^d\frac{x_i}{\theta_i}-d+1\right)^{-(a+d)}\een
for $x_i> \theta_i>0$ with $a>0$, then the density function of $X$ is $-\frac{1}{a+d}$-concave, and the probability measure is 
$-\frac{1}{a}$-concave.
\end{enumerate}

The optimal $\kappa$ for the distributions above can be found through direct computation on densities, let us also remind the reader that $\kappa$-concavity is an affine invariant.  In other words, if $X$ is $\kappa$-concave and $T$ is affine, then $TX$ is $\kappa$-concave as well, which supplies further examples through modification of the examples above. 

We will also find useful an extension of Lemma~\ref{lem:theLemma} 
to convex measures (this was obtained in \cite{BM11:it} under an additional condition,
which was removed in \cite{BFLM16}).

\begin{lem}\label{lem:comp-cvx}
Let $\kappa \in (-\infty,0]$. If $X$ is a $\kappa$-concave random vector in $\R^d$, 
then
\be\label{kconc-ub}
h(X) - h_{\infty}(X) \, \leq\  \sum_{i=0}^{d-1} \frac{1-\kappa d}{1 - \kappa i} ,
\ee
with equality for the $n$-dimensional Pareto distribution.
\end{lem}
To match notation with \cite{BM11:it} notice that $X$ being $\kappa$-concave is equivalent to $X$ having a density function that can be expressed as $\varphi^{-\beta}$, for $\beta = d - \frac 1 \kappa$ and $\varphi$ convex.

We now develop reverse R\'enyi entropy power inequalities for $\kappa$-concave measures, 
inspired by work on special cases (such as the log-concave case corresponding to $\kappa =0$ in the terminology above,
or the case of Shannon-Boltzmann entropy) in \cite{BM12:jfa, Yu08:2, BM13:goetze, BNT15}.

\subsection{Positional Reverse EPI's for R\'enyi entropies}
\label{sec:repi-pos}

The reverse Brunn-Minkowski inequality (Reverse BMI) is a celebrated result in convex geometry
discovered by V.~Milman  \cite{Mil86} (see also \cite{Mil88:2,Mil88:1,Pis89:book})
It states that given two convex bodies $A$ and $B$ in $\mathbb{R}^d$, one can find a linear 
volume-preserving map $u:\mathbb{R}^d\rightarrow\mathbb{R}^d$ such that with some absolute constant $C$, 
\begin{eqnarray}\label{inq:InvBM}
|u(A)+B|^{1/d}\le C(|A|^{1/d}+|B|^{1/d}) .
\end{eqnarray}

The EPI may be formally strengthened by using the
invariance of entropy under affine transformations of determinant $\pm 1$, 
i.e., $N(u(X)) = N(X)$ whenever $|{\rm det}(u)|=1$. Specifically,
\be\label{epi-aff}
\inf_{u_1, u_2} N(u_1(X)+u_2(Y)) \geq N(X) + N(Y),
\ee
where the maps $u_i:\RL^d \rightarrow \RL^d$ range over all 
affine entropy-preserving transformations. It was shown in \cite{BM11:cras} that
in exact analogy to the Reverse BMI,
the inequality \eqref{epi-aff} can be reversed with a constant not depending
on dimension if we restrict to log-concave distributions.
To state such results compactly, we adopt the following terminology.

\begin{defn}\label{defn:pos-repi}
For each $d\in\Nat$, let $\calM_d$ be a class of probability measures on $\R^d$,
and write $\calM=(\calM_d: d\in \Nat)$. Suppose that for every pair of independent
random variables $X$ and $Y$ whose distributions lie in $\calM_d$, 
there exist linear maps $u_1, u_2:\RL^d \rightarrow \RL^d$ of determinant 1 such that 
\be\label{eq:repi}
N_{p}\big(u_1(X) + u_2(Y)\big)\, \leq\, C_{p}\, (N_{p}(X) + N_{p}(Y)),
\ee
where $C_{p}$ is a constant that depends only on $p$ (and not on $d$ or the distributions of $X$ and $Y$).
Then we say  that a Positional Reverse $p$-EPI holds for $\calM$.
\end{defn}

\begin{thm}\label{thm:repi}\cite{BM11:cras}
Let $\calM_d^{LC}$ be the class of log-concave probability measures on $\R^d$,
and $\calM^{LC}=(\calM_d^{LC}: d\in \Nat)$. A Positional Reverse $1$-EPI holds for  $\calM^{LC}$.
\end{thm}

Specializing to uniform distributions on convex bodies, it is shown in \cite{BM12:jfa} that
Theorem~\ref{thm:repi} recovers the Reverse BMI.
Thus one may think of Theorem~\ref{thm:repi} as completing in a reverse direction 
the already extensively discussed analogy between the BMI and EPI.

Furthermore, \cite{BM12:jfa} found\footnote{Actually \cite{BM12:jfa} only proved this under the additional
condition that $\beta\geq 2d+1$, but it turns out that this condition can be dispensed with,
as explained in \cite{MWB16}.} that 
Theorem~\ref{thm:repi} can be extended to larger subclasses of the class of convex measures.

\begin{thm}\label{thm:cvx-repi}\cite{BM12:jfa}
For $\beta_0>2$, let $\calM_{d,\beta_0}$ be the class of probability measures whose densities of the form 
$f(x)=V(x)^{-\beta}$ for $x\in\mathbb{R}^d$, 
where $V: \R^d \ra (0,\infty]$ is a positive convex function and $\beta\ge \beta_0 d$. 
Then a Positional Reverse $1$-EPI holds for $\calM_{\beta_0}=(\calM_{d, \beta_0}: d\in \Nat)$.
\end{thm}

In \cite{BM13:goetze}, it is shown that a Reverse EPI is not possible over all convex measures. 

\begin{thm}\label{bm-thm:neg}\cite{BM13:goetze}
For any constant $C$, there is a convex probability 
distribution $\mu$ on the real line with a finite entropy, such that
$$
\min\{N(X+Y), N(X-Y)\} \geq C\,N(X),
$$
where $X$ and $Y$ are independent random variables distributed according to $\mu$.
\end{thm}

We have the following positional reverse $p$-R\'enyi EPI for log-concave random vectors;
this does not seem to have explicitly observed before.

\begin{thm}\label{thm:REPIREnyi}
For any $p\in (0,\infty]$, a Positional Reverse $p$-R\'enyi EPI holds for  $\calM^{LC}$. Moreover, for $p\ge 1$, 
the constant $C_{\calM, p}$ in the corresponding inequality 
does not depend on $p$.
\end{thm}

\begin{proof}
For any pair of independent log-concave random vectors $X$ and $Y$, there exist linear maps $u_1$, $u_2$: $\R^d\rightarrow \R^d$ of determinant 1, such that for all $p> 1$, by Lemma \ref{lem:monotonicity}, Theorem \ref{thm:repi} and Lemma \ref{lem:theLemma}, one has
\begin{align*}
N_p(u_1(X)+u_2(Y))
&\le N(u_1(X)+u_2(Y))\lesssim N(X)+N(Y)\\
&\lesssim N_\infty(X)+N_\infty(Y)\le N_p(X)+N_p(Y).
\end{align*}

For $p<1$, by Lemma \ref{lem:theLemma} and Lemma \ref{lem:monotonicity}, there exist a constant $C(p)$ depending solely on $p$ such that
\begin{align*}
N_p(u_1(X)+u_2(Y))&\le C(p)N(u_1(X)+u_2(Y))\le C(p)\left(N(X)+N(Y)\right)\\
&\le C(p)\left(N_p(X)+N_p(Y)\right),
\end{align*}
which provides the theorem.
\end{proof}

Later we will show that Theorem~\ref{thm:REPIREnyi} can be used to recover the functional version
of the reverse Brunn-Minkowski inequality proposed by Klartag and V.~Milman \cite{KM05}.

\subsection{Reverse $\infty$-EPI via a generalization of K.~Ball's bodies}
\label{sec:repi-ball}

\subsubsection{Busemann's theorem for convex bodies}

We first consider Bobkov's extension of K.~Ball's convex bodies associated to log-concave measures.  In this direction 
we associate a star shaped body to a density function via a generalization of the Minkowski functional of a convex body.

\begin{defn}
For a probability density function $f$ on $\mathbb{R}^d$ with the origin in the interior of the support of $f$, 
and $p \in (0,\infty)$, define $\Lambda_f^{p}: \mathbb{R}^d \to [0,\infty]$ by
\begin{align*}
 	\Lambda_f^{p}(v) = \left( \int_0^\infty f(rv) dr^p \right)^{-1/p}
\end{align*}
We will consider the class of densities $\mathcal{F}_{p}$ where $\Lambda_f^{p}(v) \in [0,\infty)$ for all $v \in \mathbb{R}^d$.  For such densities, we can associate a body defined by
\begin{align*}
 	K_f^{p} = \{ v \in\R^d: \Lambda_f^{p}(v) \leq 1 \}.
\end{align*}
\end{defn}

We can now state Bobkov's generalization \cite{Bob10} of the Ball-Busemann theorem.

\begin{thm}\label{bobkovbodies}
If $f$ is a $s$-concave density on $\mathbb{R}^d$, with $-\dth \leq s \leq 0$,
then
\begin{align}
\Lambda_f^{p}((1-t)x+ty) \leq (1-t) \Lambda_f^{p}(x) + t \Lambda_f^{p}(y) ,
\end{align}
for every $x,y\in \RL^d$ and $t\in (0,1)$, provided $0<p \leq - 1 - 1/s$. 
\end{thm}

\begin{remark}\label{triangle} \normalfont
Notice that, since $\Lambda_f^{p}$ is positive homogeneous and (by Theorem~\ref{bobkovbodies}) convex,  it necessarily satisfies the triangle inequality.  
If we add the assumption that $f$ is even, then $\Lambda_f^{p}$ defines a norm.
\end{remark}

There is remarkable utility in this type of association. In \cite{Bal88}, Ball used the fact that one can directly pass from log-concave 
probability measures to convex bodies using this method to derive an analog of Hensley's theorem \cite{Hen80} for certain log-concave measures, 
demonstrating comparability of their slices by different hyperplanes.  By generalizing this association to convex measures in \cite{Bob10},
Bobkov derived analogs of Blaschke-Santalo inequalities, the Meyer-Reisner theorem \cite{MR91:1} (this was proved independently in unpublished work, by Keith Ball, as discussed in \cite{MR91:2}) 
for floating surfaces, and Hensley's theorem for convex measures. Thus this association of convex bodies with convex measures may be seen as a way to ``geometrize" said measures.

Another application of this association of bodies to measures is to the study of so-called intersection bodies.

\begin{defn}
For any compact set $K$ in $\R^d$ whose interior contains the origin,
define $r: \mathbb{S}^{d-1} \to (0,\infty)$ by $r(\theta) = |K \cap \theta^\perp|_{d-1}$
(i.e., the volume of the $(d-1)$-dimensional slice of $K$ by the subspace orthogonal to $\theta$).  
The star-shaped body whose boundary is defined by the points $\theta r(\theta)$ 
is called the intersection body of $K$, and denoted $I(K)$.
\end{defn}

The most important fact about intersection bodies is the classical theorem of 
Busemann \cite{Bus49}.

\begin{thm} \cite{Bus49} \label{thm:BuseCoGeo} 
If $K$ be a symmetric convex body in $\R^{d}$, then $I(K)$ is a symmetric convex body as well.
\end{thm}

The symmetry is essential here; the intersection body of a non-symmetric convex body need not be
convex\footnote{There is a nontrivial way to extend the definition of intersection body to non-symmetric convex bodies
so that the new definition results in a convex body; see \cite{MR11} for details.}. 
Busemann's theorem is a fundamental result in convex geometry since it expresses a convexity property
of volumes of central slices of a symmetric convex body, whereas Brunn's theorem (an easy implication
of the BMI) asserts a concavity property of volumes of slices that are perpendicular to a given direction.

Busemann's theorem may be recast in terms of R\'enyi entropy, as implicitly recognized by K. Ball and explicitly described below.

\begin{thm}\label{thm:BuseInfo}
If $X$ is uniformly distributed on a symmetric convex body $K \subset \mathbb{R}^d$, then the mapping $M_\infty^X:\R^d \ra \R$ defined by
\[
M_\infty^X(v) = \begin{cases} 
      N_\infty^{1/2} (\langle v, X \rangle) & v \neq 0 \\
      0 & v =0 
   \end{cases}
\]
defines a norm on $\mathbb{R}^d$.
\end{thm}

Before showing that Theorems~\ref{thm:BuseCoGeo} and \ref{thm:BuseInfo} are equivalent, we need to recall the definition of the Minkowski functional.

\begin{defn} \label{minkfunc}
For a convex body $L$ in $\R^d$ containing the origin, define $\rho_L:\R^d \to [0,\infty)$ by
\begin{align*} 
\rho_L(x) = \inf \{ t \in (0,\infty) : x \in tL\}.
\end{align*} 
\end{defn}

It is straightforward that $\rho_L$ is positively homogeneous (i.e., $\rho_L(ax) =a\rho_L(x)$ for $a>0$) and convex.  
When $L$ is assumed to be symmetric, $\rho_L$ defines a norm.

\begin{proof}[Proof of Theorem \ref{thm:BuseCoGeo} $\Leftrightarrow$ Theorem \ref{thm:BuseInfo}]
Let $K$ be a symmetric convex body and without loss of generality take $|K|=1$. 
Let $X = X_K$ denote a random variable distributed uniformly on $K$.  

For a unit vector $\theta\in\mathbb{S}^{d-1}$, 
as the pushforward of a symmetric log-concave measure under the linear map $x\mapsto\langle\theta, x\rangle$,
the distribution of the real-valued random variable $\langle\theta, X\rangle$  is symmetric and log-concave.
Denoting the symmetric, log-concave density of $\langle\theta, X\rangle$ by $f_\theta$, we see that the mode
of $f_\theta$ is 0, and consequently,
\[
N^{1/2}_\infty(\langle\theta, X\rangle) = \frac{1}{f_\theta(0)}= \frac{1}{|K \cap \theta^\perp|_{d-1}} = \frac{1}{r(\theta)} .
\]
By the definition of $I(K)$, we have $\rho_{I(K)}( r(\theta) \theta) = 1$.  Thus, for any $\theta\in\mathbb{S}^{d-1}$, 
\[
\rho_{I(K)}(\theta)= \rho_{I(K)}\bigg(\frac{r(\theta)\theta}{r(\theta)}\bigg) = \frac{1}{r(\theta)} = M_\infty^X (\theta) .
\]
By homogeneity, this immediately extends to $\R^d$,
establishing our result and also a pleasant duality; 
up to a constant factor, the Minkowski functional associated to $I(K)$ is a R\'enyi entropy power of the projections of $X_K$.
\end{proof}

\subsubsection{A Busemann-type theorem for measures}

Theorem \ref{thm:BuseInfo} is a statement about $\infty$-R\'enyi entropies associated to 
a $1/d$-concave random vector $X$ (see Example 1 after Theorem~\ref{thm:k-kd}).
It is natural to wonder if Busemann's theorem can be extended to other $p$-R\'enyi entropies and more general classes of measures.  

In \cite{BNT15}, Ball-Nayar-Tkocz also give a simple argument, essentially going back to \cite{Bal88}, that 
the information-theoretic statement of Busemann's theorem (namely Theorem~\ref{thm:BuseInfo}) extends to log-concave measures.  
Interpreting in the language of Borell's $\kappa$-concave measures, \cite{BNT15} extends Theorem~\ref{thm:BuseInfo} to measures that are 
$\kappa$-concave with $\kappa \geq 0$.  In what follows, we use the same argument as \cite{BNT15} to prove that 
Busemann's theorem can in fact be extended to all convex measures by invoking Theorem~\ref{bobkovbodies}.

\begin{thm} \label{prop:buse} 
Let $\kappa \in [-\infty, 1/2]$. If $(U,V)$ is a symmetric $\kappa$-concave random vector in $\mathbb{R}^2$, then
\[
e^{h_\infty(U+V)} \leq e^{h_\infty(U)} + e^{h_\infty(V)} .
\]
\end{thm}
\begin{proof}
It is enough to prove the result for the weakest hypothesis $\kappa = -\infty$. We let $\phi$ denote the density function of $(U,V)$ so that
\begin{align*}
	U+V \sim w(x) &= \int_\mathbb{R} \phi(x-t,t) dt
		\\
	U\sim u(x) &= \int_\mathbb{R} \phi(x,t) dt
		\\
	V\sim v(x) &= \int_\mathbb{R} \phi(t,x) dt.
\end{align*}
Since symmetry and the appropriate concavity properties of the densities forces the maxima of $u, v, w$ to occur at $0$, 
\begin{align*}
 	\frac 1 {\|w\|_\infty} 
 		&= 
 			\frac 1 {w(0)}
 			\\
 		&= 
 			\left( \int_\mathbb{R} \phi(-t,t) dt \right)^{-1}
 			\\
 		&=
 			\left( 2 \int_0^\infty \phi(t(e_2 -e_1)) dt \right)^{-1}
 			\\
 		&=
 			\frac 1 2 \Lambda_\phi^{1}(e_2-e_1)
 			\\
 		&\leq
 			\frac 1 2 \left( \Lambda_\phi^{1}(e_2) + \Lambda_\phi^{1}(e_1) \right)
 			\\
 		&=
 			\left( 2 \int_0^\infty \phi(0,t) dt \right)^{-1} + \left( 2 \int_0^\infty \phi(t,0) dt \right)^{-1}
 			\\
 		&=
 			\frac 1 {u(0)} + \frac 1 {v(0)}
 			\\
 		&=
 			\frac 1 {\|u\|_\infty} + \frac 1 {\|v\|_\infty} ,
\end{align*}
where the only inequality follows from Theorem \ref{bobkovbodies} with $a=1$ and $p=1 = n - 1 - 1/\kappa$.
By definition of $h_\infty$, we have proved the desired inequality.
\end{proof}

As a nearly immediate consequence we have Busemann's theorem for convex measures.
\begin{cor} For $\kappa \in [-\infty,\frac 1 d]$, if $X$ is symmetric and $\kappa$-concave the function 
\[
M_\infty^X(v) = \begin{cases} 
      N_\infty^{1/2} (\langle v, X \rangle) & v \neq 0 \\
      0 & v =0 
   \end{cases}
\]
defines a norm.
\end{cor}
\begin{proof}
As we have observed $M = M_\infty^X$ is homogeneous.  To prove the triangle inequality take vectors $u,v \in \mathbb{R}^d$ and define $(U,V) = (\langle X, u \rangle,\langle X, v \rangle)$, so that $U+V = \langle X, u+v \rangle$.  Notice that $(U,V)$ is clearly symmetric and as the affine pushforward of a $\kappa$-concave measure, is thus $\kappa$-concave as well.  Thus by Theorem~\ref{prop:buse} we have
\[
e^{h_\infty(U+V)} \leq e^{h_\infty(U)} + e^{h_\infty(V)}.
\]
But this is exactly
\[
N_\infty^{1/2}(\langle X,u+v \rangle) \leq N_\infty^{1/2}(\langle X,u \rangle)  + N_\infty^{1/2}(\langle X,v \rangle),
\]
which is what we sought to prove.
\end{proof}

\subsubsection{Busemann-type theorems for other R\'enyi entropies}

While the above extension deals with general measures, a further natural question relates to more general entropies.
Ball-Nayar-Tkocz \cite{BNT15} conjecture that the Shannon entropy version holds for log-concave measures.

\begin{conj}\cite{BNT15}\label{conj:bnt1} 
When $X$ is a symmetric log-concave vector in $\mathbb{R}^d$ then the function
\[
M_1^X(v) = \begin{cases} 
      N_1^{1/2}(\langle v, X \rangle) & v \neq 0 \\
      0 & v =0 
   \end{cases}
\]
defines a norm on $\mathbb{R}^d$.
\end{conj}

As the homogeneity of $M$ is immediate, the veracity of the conjecture depends on proving the triangle inequality
\[
e^{h_1(\langle u +v, X \rangle)} \leq e^{h_1(\langle v, X \rangle)} + e^{h_1(\langle u, X \rangle)} ,
\]
which is easily seen to be equivalent to the following modified Reverse EPI for symmetric log-concave measures on $\mathbb{R}^2$. 

\begin{conj}\cite{BNT15}\label{conj:bnt2}
For a symmetric log-concave random vector in $\R^2$, with coordinates $(U,V)$,
\ben
N_1^{1/2}(U+V) \leq N_1^{1/2}(U) + N_1^{1/2}(V).
\een
\end{conj}

Towards this conjecture, it is proved in \cite{BNT15}  that $e^{\alpha h_1(U+V)} \leq e^{\alpha h_1(U)} + e^{\alpha h_1(V)}$ when $\alpha=1/5$.
By extending the approach used by \cite{BNT15}, we can obtain a family of Busemann-type results for $p$-R\'enyi entropies.

\begin{thm}\label{thm:BNT-p}
Fix $p\in [1,\infty]$. There exists a constant $\alpha_p>0$ which depends only on the parameter $p$, such that for a 
symmetric log-concave random vector $X$ in $\RL^d$ and two vectors $u,v\in\RL^d$, we have
\ben
e^{\alpha_p h_p(\langle u+v,X\rangle)}\le e^{\alpha_p h_p(\langle u,X\rangle)}+e^{\alpha_p h_p(\langle v,X\rangle)} .
\een
Equivalently, for a symmetric log-concave random vector $(X,Y)$ in $\mathbb{R}^2$ we have
\ben
e^{\alpha_p h_p(X+Y)}\le e^{\alpha_p h_p(X)}+e^{\alpha_p h_p(Y)} .
\een
In fact, if $p\in [1,\infty)$, one can take $\alpha_p$ above to be the unique positive solution $\alpha$ of
\be\label{inq:crucial2}
p^{\frac{\alpha}{p-1}}= \theta_p^{\alpha}+(1-\theta_p)^{\alpha} ,
\ee
where
\ben
\theta_p:=\bigg(\frac{\log p}{p-1}\bigg) \cdot \frac{1}{2(e+1)[2pe^2+(4p+1)e+1]}  ,
\een
with the understanding that the $p=1$ case is understood by continuity (i.e., the left side of equation \eqref{inq:crucial2}
is $e^\alpha$ in this case, and the pre-factor $\frac{\log p}{p-1}$ in $\theta_p$ is replaced by 1).
\end{thm}

\begin{remark}
If $p< \infty$, then $\theta_p>0$, and on the other hand, trivially $\theta_p < \frac{1}{2(1+e)} <1$.
Denote the left and right sides of the equation \eqref{inq:crucial2} by $L_p(\alpha)$ and $R_p(\alpha)$ respectively.
Then $1=L_p(0) < R_p(0)=2$, and since $p^{1/(p-1)}>1$ for $p\in [1,\infty)$, we also have $\infty=\lim_{\alpha\ra\infty} L_p(\alpha)> \lim_{\alpha\ra\infty} R_p(\alpha)=0$. 
Since $L_p$ and $R_p$ are continuous functions of $\alpha$, equation \eqref{inq:crucial2} must have a positive solution $\alpha_p$.
Moreover, since $L_p$ is an increasing function and $R_p$ is a decreasing function,  there must be a unique positive solution $\alpha_p$.
In particular, easy simulation gives $\alpha_1\approx 0.240789 > 1/5$, and simulation also shows that the unique solution
$\alpha_p$ is non-decreasing in $p$. Consequently it appears that for any $p$, one can replace $\alpha_p$ in the above theorem by $1/5$.
\end{remark}

Since Theorem~\ref{thm:BNT-p} is not sharp, and the proof involves some tedious and unenlightening calculations,
we do not include its details. We merely mention some analogues of the steps used by \cite{BNT15} to prove the
case $p=1$. As done there, one can ``linearize'' the desired inequality to obtain the following equivalent form: 
if $(X,Y)$ is a symmetric log-concave vector in $\mathbb{R}^2$ with $h_p(X)=h_p(Y)$, then for every $\theta\in [0,1]$,
\ben
h_p(\theta X+(1-\theta)Y)\le h_p(X)+\frac{1}{\alpha_p}\log \left(\theta^{\alpha_p}+(1-\theta)^{\alpha_p}\right) .
\een
To prove this form of the theorem, it is convenient as in \cite{BNT15} to divide into cases where $\theta$ is ``small''
and ``large''. For the latter case, the bound
\begin{align*}
e^{h_p(X+Y)}\le e^{h_\infty(X+Y)+\frac{\log p}{p-1}}=p^{1/(p-1)}\left(e^{h_\infty(X)}+e^{h_\infty(Y)}\right)
\le p^{1/(p-1)}\left(e^{h_p(X)}+e^{h_p(Y)}\right) ,
\end{align*}
easily obtained  by combining Lemmata~\ref{lem:monotonicity} and ~\ref{lem:theLemma}, suffices.
The former case is more involved and relies on proving the following extension of \cite[Lemma 1]{BNT15}:
If $w:\mathbb{R}^2\rightarrow\mR_+$ is a symmetric log-concave density,
and we define $f(x):=\int w(x,y)dy$ and $\gamma=\int w(0,y)dy/\int w(x,0)dx$, then
\ben
\frac{\int\int -f(x)^{p-2}f'(x)yw(x,y)dxdy}{\int f(x)^pdx}\le \left(2e(e+2)+\frac{e+1}{p}\right)\gamma .
\een

Staring at Theorem~\ref{prop:buse} and Conjecture~\ref{conj:bnt2}, and given that one would expect to be able to interpolate
between the $p=1$ and $p=\infty$ cases, it is natural to pose the following conjecture that would subsume all of the results and
conjectures discussed in this section.

\begin{conj}\label{conj:gen-bnt}
Fix $\kappa \in [-\infty,\frac 1 d]$. 
For a symmetric $\kappa$-concave random vector in $\R^2$, with coordinates $(U,V)$,
it holds for any $p\in [1,\infty]$ that
\ben
N_p^{1/2}(U+V) \leq N_p^{1/2}(U) + N_p^{1/2}(V) ,
\een
whenever all these quantities are finite.
Equivalently, when $X$ is a symmetric $\kappa$-concave random vector in $\mathbb{R}^d$, then for any given $p\in [1,\infty]$,
the function
\[
M_p^X(v) = \begin{cases} 
      N_p^{1/2}(\langle v, X \rangle) & v \neq 0 \\
      0 & v =0 
   \end{cases}
\]
defines a norm on $\mathbb{R}^d$ when it is finite everywhere.
\end{conj}

Given the close connection of the $p=\infty$ case with intersection bodies and Busemann's theorem,
one wonders if there is a connection between the unit balls of the conjectured norms $M_p^X$ in Conjecture~\ref{conj:gen-bnt} on the one hand,
and the so-called $L_p$-intersection bodies that arise in the dual $L_p$ Brunn-Minkowski theory (see, e.g., Haberl \cite{Hab08}) on the other.

After the first version of this survey was released, Jiange Li (personal communication) has verified that
Conjecture~\ref{conj:gen-bnt} is true when $p=0$ (with arbitrary $\kappa$) and when $p=2$ (with $\kappa=0$, i.e., in the log-concave case).

\subsection{Reverse EPI via R\'enyi entropy comparisons}
\label{sec:repi-bm13}

The Rogers-Shephard inequality \cite{RS57} is a classical and influential inequality in Convex Geometry.
It states that for any convex body $K$ in $\R^d$,
\be\label{inq:RS}
|K-K|\le {2d \choose d}\mbox{Vol}(K)
\ee
where $K-K:=\{x-y: x, y\in K\}$. Since ${2d \choose d}< 4^d$, this implies that
$|K-K|^{1/d}<4|K|^{1/d}$, complementing the fact that $|K-K|^{1/d}\geq 2|K|^{1/d}$ by the BMI. In particular,
the Rogers-Shephard inequality may be thought of as a Reverse BMI.
In this section, we discuss integral and entropic liftings of the Rogers-Shephard inequality.

An integral lifting of the Rogers-Shephard inequality was developed by Colesanti \cite{Col06} (see also \cite{AGJV16, AEFO15}).
For a real non-negative function $f$ defined in $\mathbb{R}^d$, define the difference function $\Delta f$ of $f$,
\be\label{eqn:defDiff}
\Delta f(z):=\sup\{\sqrt{f(x)f(-y)}: x,y\in\mathbb{R}^d, \, \frac{1}{2}(x+y)=z\}
\ee
It is proved in \cite{Col06} that if $f:\R^d \ra [0,\infty)$ is a log-concave function, then
\be\label{inq:fnlRS}
\int_{\mathbb{R}^d}\Delta f(z)dz\le 2^d\int_{\mathbb{R}^d}  f(x)dx ,
\ee
where the equality is attained by multi-dimensional exponential distribution. 

On the other hand, an entropic lifting  of the Rogers-Shephard inequality was developed by \cite{BM13:goetze}.
We develop an extension of their argument and result here. In order to state it, we need to recall the notion
of relative entropy between two distributions: if $X, Y$ have densities $f, g$ respectively, then
\ben
D(X\|Y)=D(f\|g):=\int_{\R^d} f(x) \log \frac{f(x)}{g(x)} dx
\een
is the relative entropy between $X$ and $Y$. By Jensen's inequality, $D(X\|Y)\geq 0$, with
equality if and only if the two distributions are identical.

\begin{lem}\label{lem:diff}
Suppose $(X, Y)\in \R^d\times \R^d$ has a $\kappa$-concave distribution, with $\kappa<0$. 
If $X$ and $Y$ are independent, then
\ben
h(X-Y)\leq \min\{ h(X) +D(X\|Y) , h(Y) +D(Y\|X) \} +  \sum_{i=0}^{d-1} \frac{1-\kappa d}{1 - \kappa i}  .
\een 
\end{lem}

\begin{proof}
By affine invariance, the distribution of $X-Y$ is $\kappa$-concave, so that 
one can apply Lemma~\ref{lem:comp-cvx} to obtain
\ben\begin{split}
h(X-Y)&\leq \log \|f\|_{\infty}^{-1} +  \sum_{i=0}^{d-1} \frac{1-\kappa d}{1 - \kappa i}  \\
&\leq \log f(0)^{-1} +  \sum_{i=0}^{d-1} \frac{1-\kappa d}{1 - \kappa i} .
\end{split}\een
Denoting the marginal densities of $X$ and $Y$ by $f_1$ and $f_2$ respectively, we have $f(0)=\int_{\RL^d} f_1(x) f_2(x) dx$, and hence
\ben\begin{split}
h(X-Y)
&\leq  -\log \int_{\RL^d} f_1(x) f_2(x) dx +  \sum_{i=0}^{d-1} \frac{1-\kappa d}{1 - \kappa i}  \\
&\leq \int_{\RL^d} f_1(x) [-\log f_2(x)] dx +  \sum_{i=0}^{d-1} \frac{1-\kappa d}{1 - \kappa i}  \\
&= h(X) +D(X\|Y) +  \sum_{i=0}^{d-1} \frac{1-\kappa d}{1 - \kappa i}  .
\end{split}\een
Clearly the roles of $X$ and $Y$ here are interchangeable, yielding the desired bound.
\end{proof}

In the case where the marginal distributions are the same, Lemma~\ref{lem:diff} reduces as follows.

\begin{thm}\label{thm:main}
Suppose $(X, Y)\in \R^d\times \R^d$ has a $\kappa$-concave distribution, with $\kappa<0$. 
If $X$ and $Y$  are independent and identically distributed, then
\ben
N(X-Y) \leq C_{\kappa} N(X) ,
\een
where 
\ben
C_{\kappa}= \exp\bigg\{ \frac{2}{d} (1-d\kappa) \sum_{j=0}^{d-1} \frac{1}{1-j\kappa} \bigg\}.
\een
\end{thm}

As $\kappa\ra 0$, we recover the fact, obtained in \cite{BM13:goetze},  that $N(X-Y) \leq e^2 N(X)$ for $X, Y$ i.i.d. with log-concave marginals.
We believe that this statement can be tightened, even in dimension 1. 
Indeed, it is conjectured in \cite{MK15} that for $X, Y$ i.i.d. with log-concave marginals,
\ben
N(X-Y)\leq 4 N(X)
\een
is the tight entropic version of Rogers-Shepard in one dimension,
with equality for the one-sided exponential distribution.

\subsection{Reverse R\'enyi EPI via Convex Ordering}
\label{sec:repi-order}

\subsubsection{Convex ordering and entropy maximization}

In this section, we build on an elegant approach of Y.~Yu \cite{Yu08:2}, who obtained inequalities for R\'enyi entropy 
of order $p \in (0,1]$ for i.i.d. log-concave measures under stochastic ordering assumptions.  In particular, 
we achieve extensions to $\kappa$-concave measures with $\kappa <0$ and impose
weaker distributional symmetry assumptions,  and observe that the resulting inequalities 
may be interpreted as Reverse EPI's.  

\begin{lem}\label{lem:criterion}
Let  $X \sim f$, $Y \sim g$ be random vectors on $\R^d$. In order to prove
\[
h_p(X) \geq h_p(Y),
\]
it suffices to prove 
\begin{align} 
\label{yu01}
\mathbb{E} f^{p-1}(X) &\geq \mathbb{E} f^{p-1}(Y) , \quad \text{ if } p \in (0,1) ,\\
 \label{yupos}
\mathbb{E} f^{p-1}(X) &\leq \mathbb{E} f^{p-1}(Y), \quad \text{ if } p \in (1,\infty) ,\\
 \label{yu0}
-\mathbb{E} \log f(X) &\geq -\mathbb{E} \log f(Y) , \quad \text{ if } p =1 .
\end{align}
\end{lem}
\begin{proof}
Notice that the expressions in the hypothesis for $p \neq 1$ can be re-written as $\mathbb{E}f^{p-1}(X) = \int_{\mathbb{R}^d} f^{p-1}(x) f(x) dx$ and $\mathbb{E}f^{p-1}(Y)=\int_{\mathbb{R}^d} f^{p-1}(x) g(x) dx$.  For $p \in (0,1)$,
\ben\begin{split}
	\int f^p dx 
		&= \left( \int f^{p-1} f \right)^{p}\left( \int f^{p} \right)^{1-p} \\ 
		&\geqa \left( \int f^{p-1} g \right)^{p}\left( \int f^{p} \right)^{1-p} \\ 
		&\geqb \int g^p dx , 
\end{split}\een
where (a) is from applying the hypothesis and (b) is by H\"older's inequality (applied in the probability space $(\RL^d, g\, dx)$). Inequality \eqref{yu01}
follows from the fact that $(1-p)^{-1}\log x$ is order-preserving for $p \in (0,1)$.

When $p \in (1,\infty)$,
\ben\begin{split}
	\int f^p dx &= \left( \int f^{p-1} f \right)^{p}\left( \int f^{p} \right)^{1-p} \\
		&\leq \left( \int f^{p-1} g \right)^{p}\left( \int f^{p} \right)^{1-p} \\ 
		&\leqc \int g^p dx, 
\end{split}\een
where H\"older's inequality is reversed for $p \in (1,\infty)$ accounting for (c).
Inequality \eqref{yupos} follows   since $(1-p)^{-1} \log x$ is order-reversing for such $p$.

In the case $p =1$, we use the hypothesis and then Jensen's inequality to obtain,
\begin{align*}
h(X)
	&=
		-\mathbb{E}\log f(X)
		\\
	&\geq
		-\mathbb{E} \log f(Y)
		\\
	&\geq
		-\mathbb{E} \log g(Y)
		\\
	&=
		h(Y),
\end{align*}
which yields inequality \eqref{yu0} and completes the proof of the lemma.
\end{proof}

Of the observations in Lemma~\ref{lem:criterion}, \eqref{yu01} and \eqref{yu0} were used in \cite{Yu08:2};
we add \eqref{yupos}, which is relevant to Reverse EPI's for $\kappa$-concave measures with $\kappa>0$.

We recall the notion of convex ordering for random vectors.

\begin{defn} \normalfont
For random variables $X, Y$ taking values in a linear space $V$, we say that $X$ dominates $Y$ in the convex order, written 
$X \geq_{cx} Y$, if 
$\mathbb{E}\varphi(X) \geq \mathbb{E}\varphi(Y)$
for every  convex and continuous function $\varphi:V\ra\R$.
\end{defn}

We need a basic lemma relating supports of distributions comparable in the convex ordering.

\begin{lem} \label{zerocase}
Given random vectors $X \sim f$ and $Y \sim g$ such that $Y \leq_{cx} X$, if $\mbox{supp}(f)$ is a convex set, then $\mbox{supp}(g)\subset \mbox{supp}(f)$. 
\end{lem}
\begin{proof}
Take $\rho$ to be the Minkowski functional (Definition \ref{minkfunc}) associated to $\mbox{supp}(f)$ and then define 
\[
\varphi(x) = \max \{\rho(x)-1, 0\}.
\]
As the maximum of two convex functions, $\varphi$ is convex.  Also observe that $\varphi$ is identically zero on $\mbox{supp}(f)$ while strictly positive on the compliment.  By the ordering assumption
\ben
0 \leq \mathbb{E}(\varphi(Y))\le \mathbb{E}(\varphi(X)) = 0.
\een
Thus $\mathbb{E}(\varphi(Y))=0$, which implies the claim.
\end{proof}

We can now use convex ordering as a criterion to obtain a maximum entropy property of convex measures under certain conditions.

\begin{thm} \label{yulemma}
Let $X$ and $Y$ be random vectors in $\mathbb{R}^d$, with $X$ being 
$\kappa$-concave for some $\kappa \in (-\infty, \frac 1 d]$.
If $X \geq_{cx} Y$, then
\begin{align*}
h_p(X) \geq h_p(Y)
\end{align*}
for $0 \leq p \leq \kappa/(1-d\kappa) + 1 $.
\end{thm}

\begin{proof}
Recall that $X$ is $\kappa$-concave if and only if it admits a $s_{\kappa,d}$-concave density $f$ on its support, with $s_{\kappa,d} = \kappa/(1-d\kappa)$.  Thus it follows that for $a \leq s_{\kappa,d}$, $f^a$ is a convex function, (resp. concave) for $a<0$ (resp. $a>0$). Our hypothesis is simply that that $p - 1 \leq s_{\kappa,d}$. 

For $p <1$ we can apply the convex ordering to necessarily convex function $f^{p-1}$, as $\mathbb{E}f^{p-1}(X) \geq \mathbb{E}f^{p-1}(Y)$ and apply Lemma \ref{lem:criterion} under the hypothesis \eqref{yu01}.

When $p >1$ the proof is the same as the application of convex ordering to the concave function $f^{p-1}$ will reverse the inequality to attain $\mathbb{E}f^{p-1}(X) \leq \mathbb{E}f^{p-1}(Y)$ and then invoking Lemma \ref{lem:criterion} under hypothesis \eqref{yupos} will yield the result.

To consider $p =1$, $X$ must be at least log-concave, in which case we can follow \cite{Yu08:2} exactly.  This amounts to applying convex ordering to $-\log f$ and Lemma \ref{lem:criterion} a final time.

After recalling that the support of a $\kappa$-concave measure is a convex set, the $p =0$ case follows from Lemma~\ref{zerocase}.
\end{proof}

Theorem~\ref{yulemma} extends a result of Yu \cite{Yu08:2}, who shows that for $X$ log-concave, 
$h_p(X) \geq h_p(Y)$ for $0 < p \leq 1$ when $X \geq_{cx} Y$. Observe that as $\kappa$ approaches $1/d$,
the upper limit of the range of $p$ for which Theorem~\ref{yulemma} applies approaches $\infty$.

Some care should be taken to interpret Theorem~\ref{yulemma} and the entropy inequalities to come.  For example, 
the $t$-distribution (see Example 4 after Theorem~\ref{thm:k-kd}) does not have finite $p$-R\'enyi entropy when $p \leq \frac{d}{\nu+d}$ 
and hence the theorem only yields non-trivial results on the interval $(\frac{d}{\nu+d},1-\frac{1}{\nu+d}]$.  
Notice that in the important special case where $X$ is Cauchy,  corresponding to $\nu = 1$, this interval is empty;
thus Theorem~\ref{yulemma} fails to give a maximum entropy characterization of the Cauchy distribution
(which is of interest from the point of view of entropic limit theorems).

\begin{defn} \normalfont
We say that a family of random vectors $\{ X_1,\dots, X_n \}$ is exchangeable when 
$(X_{\sigma(1)}, \dots, X_{\sigma(n)})$ and $(X_1, \dots, X_n)$ are identically distributed for any 
permutation $\sigma$  of $\{1, \dots, n\}$.
\end{defn}

\subsubsection{Results under an exchangeability condition}

Let us also remind the reader of the notion of majorization for $a,b \in \mathbb{R}^n$. First we recall that 
a square matrix is doubly stochastic if its row sums and column sums are all equal to 1.

\begin{defn} \normalfont
For vectors $a, b \in \mathbb{R}^n$, we will write $b \prec a$ (and say that $b$ is majorized by $a$) if there exists a doubly stochastic matrix $M$ such that $Ma = b$.  
\end{defn}

There are several equivalent formulations of this notion that are well studied (see, e.g., \cite{Sim11:book}),
but we will not have use for them. Note that if ${\bf 1}$ is the vector with all coordinates equal to 1,
then ${\bf 1}^T (M a)= ({\bf 1}^T M) a= {\bf 1}^T a$,
implying that $b\prec a$ can only hold if the sum of coordinates
of $a$ equals the sum of coordinates of $b$.

\begin{lem} \label{exchange}
Let $X_1, \dots, X_n$  be exchangeable random variables taking values in a real vector space $V$,
and let  $\phi:V^n \to \mathbb{R}$ be a convex function symmetric in its coordinates. If $b \prec a$,
\[
\mathbb{E}\varphi(a_1 X_1, \dots, a_n X_n) \geq \mathbb{E} \varphi( b_1 X_1, \dots, b_n X_n).
\]
\end{lem} 

\begin{proof}
Since every doubly stochastic matrix can be written as the convex combination of permutation matrices by the 
Birkhoff von-Neumann theorem (see, e.g., \cite{Sim11:book}), 
we can write $b \prec a$ as $b = (\sum_\sigma \lambda_\sigma P_\sigma)a$ where $\lambda_i \in [0,1]$ with $\sum_\sigma \lambda_\sigma =1$ and $P_\sigma$ is a permutation matrix.   We compute
\begin{align*}
\mathbb{E} \phi(b_1 X_1, \dots,  b_n X_n) 
	&=
		\mathbb{E} \phi\left( (\sum_\sigma \lambda_\sigma P_\sigma a)_1 X_1, \dots, (\sum_\sigma \lambda_\sigma P_\sigma a)_n X_n \right)
		\\
	&\leq
		\sum_\sigma \lambda_\sigma \mathbb{E} \phi( a_{\sigma(1)}X_1, \dots, a_{\sigma(n)} X_n)
		\\
	&= 
		\sum_\sigma \lambda_\sigma \mathbb{E} \phi( a_{\sigma(1)} X_{\sigma(1)}, \dots, a_{\sigma(n)} X_{\sigma(n)})
		\\
	&=
		\sum_\sigma \lambda_\sigma \mathbb{E} \phi(a_1 X_1, \dots, a_n X_n)
		\\
	&=
		\mathbb{E} \phi(a_1 X_1, \dots, a_n X_n) ,
\end{align*}
where the steps are justified-- in order-- by definition, convexity, exchangeability, coordinate symmetry, and then algebra. 
\end{proof}

\begin{thm}\label{thm:nonindep}
Let $X = (X_1,\dots, X_n)$ be an exchangeable collection of $d$-dimensional random vectors. Suppose $b \prec a$ and that $a_1X_1+\cdots +a_nX_n$ 
has a $s$-concave density. Then for any $p \in [0,s+1]$,
\[
h_p(b_1 X_1 + \cdots + b_n X_n)  \leq h_p (a_1 X_1 + \cdots + a_n X_n).
\]
\end{thm}

\begin{proof}
Let $f$ denote the $s$-concave density function of $a_1 X_1 + \cdots + a_n X_n$.  Thus for $p <1$ (resp. $p >1$) the function
\[
\varphi(x_1, \dots, x_n) = f^{p-1}(x_1 + \cdots + x_n)
\]
is convex (resp. concave) and clearly symmetric in its coordinates, hence by Lemma \ref{exchange} 
\begin{align*}
&\mathbb{E}\phi(b_1 X_1, \dots, b_n X_n) \leq \mathbb{E}\phi(a_1 X_1, \dots, a_n X_n), \\
&\left( \mbox{ resp. } \mathbb{E}\phi(b_1 X_1, \dots, b_n X_n) \geq \mathbb{E}\phi(a_1 X_n, \dots, a_n X_n) \right).
\end{align*}
But this is exactly,
\begin{align*}
&\mathbb{E}f^{p-1}(b_1 X_1, \dots, b_n X_n) \leq \mathbb{E}f^{p-1}(a_1 X_n, \dots, a_n X_n), \\
& \left( \mbox{ resp. } \mathbb{E}f^{p-1}(b_1 X_1, \dots, b_n X_n) \geq \mathbb{E}f^{p-1}(a_1 X_n, \dots, a_n X_n) \right),
\end{align*}
and thus by Lemma \ref{lem:criterion},
\[
h_p(b_1 X_1 +  \cdots + b_n X_n) \leq h_p(a_1 X_n + \cdots + a_n X_n).
\]
The case $p=1$ is similar by setting 
\[
\varphi(x_1, \dots, x_n) = -\log  f(x_1 + \cdots + x_n) ,
\]
and applying Lemma \ref{exchange} and Lemma \ref{lem:criterion}.
\end{proof}

\begin{defn} \normalfont
For $\Omega \subseteq \mathbb{R}^d$, we define a function $\phi: \Omega \to \mathbb{R}$ to be Schur-convex in the case that for any $x,y \in \Omega$ with $x \prec y$ we have  $\phi(x) \leq \phi(y)$.
\end{defn}

\begin{cor} \label{Cor:Schur}
Suppose $X = (X_1, \dots X_n)$ is an exchangeable collection of random vectors in $\R^d$,
with $X$ being $\kappa$-concave. Let $\Delta_n=\{\theta \in [0,1]^n: \sum_{i=1}^n \theta_i = 1 \}$
be the standard simplex, and define the function $\Phi_{X,p}: \Delta_n \to \R$  by 
\[
\Phi_{X,p}(\theta) = h_p( \theta_1 X_1 + \cdots + \theta_n X_n) .
\]
For $p \in [0, s_{\kappa,d}+1]$, $\Phi_{X,p}$ is a Schur-convex function.
In particular, $\Phi_{X,p}$ is maximized by the standard basis elements $e_i$, and minimized by $(\frac 1 n, \dots, \frac 1 n)$.
\end{cor}

\begin{proof}
If $X$ is $\kappa$-concave, then by affine invariance $\theta_1 X_1 + \cdots + \theta_n X_n$ is $\kappa$-concave, and hence Theorem \ref{thm:nonindep} applies.  The extremizers are identified by observing that for any $\theta$ in the simplex
\[
(1/n,, \cdots, 1/n)  \prec \theta \prec e_i,
\]
and the corollary follows.
\end{proof}

Of course, using the standard simplex  is only a matter of normalization; analogous results are easily obtained 
by setting $\sum_i \theta_i$ to be any positive constant.

Let us remark that when the coordinates of $X_i$ are assumed to be independent, then $X$ is log-concave if and only if 
each $X_i$ each log-concave.  As a consequence we recover in the $\kappa=0$ and $p \le 1$ case, 
the theorem of Yu in \cite{Yu08:2}.

\begin{thm}\cite{Yu08:2}\label{thm:Yu'sThm}
Let $X_1,\cdots,X_n$ be i.i.d. log-concave random vectors in $\mathbb{R}^d$. Then the function $a \mapsto  h_p(a_1 X_1 + \dots + a_n X_n)$ is Schur-convex on the simplex for $p \in (0,1]$.
\end{thm}

\subsubsection{Results under an assumption of identical marginals}

We now show that  the exchangeability hypothesis can be loosened in Corollary~\ref{Cor:Schur}.

\begin{thm} \label{thm:maxpres}
Let $X = (X_1,\dots, X_n)$ be a collection of $d-$dimensional random vectors with $X_i$ identically distributed 
and $\kappa$-concave. For $p \in [0, s_{\kappa,d}+1]$, the function $\Phi_{X,p}$ defined in 
Corollary~\ref{Cor:Schur} satisfies
\ben
	\Phi_{X,p}(a) \leq \Phi_{X,p}(e_i).
\een 
Stated explicitly, for $a \in \Delta_n$, we have
\[
	h_p(a_1 X_1 + \cdots + a_n X_n) \leq h_p(X_1).
\]
\end{thm}

\begin{proof}
Let $f$ be the density function of $X_1$ and $a \in \Delta_n$. If $p<1$, by Lemma \ref{lem:criterion}, it suffices to prove that
\ben
\mathbb{E}f^{p-1}(a_1 X_1 + \cdots + a_n X_n)
\le \mathbb{E}f^{p-1}(X_1).
\een
Since $f$ is a $s_{\kappa,d}$-concave function and $p-1\leq s_{\kappa,d}$, $f$ is also $(p-1)$-concave, which means that
$f^{p-1}$ is  convex. Consequently, we have
\begin{align*}
\mathbb{E}f^{p-1}(a_1 X_1 + \cdots + a_n X_n)
	&\leq
		a_1 \mathbb{E}f^{p-1}(X_1) + \cdots + a_n \mathbb{E}f^{p-1}(X_n)
		\\
	&=
		\mathbb{E}f^{p-1}(X_1) ,
\end{align*}
where the equality is by the fact that $X_i$ are identically distributed. The cases of $p>1$ and $p=1$ follow similarly. 
\end{proof}

\begin{cor}\label{cor:ident-tri}
Suppose $X_1$, $X_2$, $\cdots$, $X_n$ are identically distributed and $\kappa$-concave. If $p \in [0,s_{\kappa,d}+1]$, 
we have the triangle inequality
\ben
N_p^{1/2}\left(\sum_{i=1}^nX_i\right)\le \sum_{i=1}^nN_p^{1/2}\left(X_i\right).
\een
Moreover, for any $p> s_{\kappa,d}+1$,
\ben
N_p^{1/2}\left(\sum_{i=1}^nX_i\right)\le \frac{(s_{\kappa,d} + 1)^{1/s_{\kappa,d}}}{p^{1/(p-1)}}\sum_{i=1}^nN_p^{1/2}\left(X_i\right).
\een
\end{cor}
\begin{proof}
We have, by Theorem \ref{thm:maxpres}, for $p \in [0,s_{\kappa,d}+1]$,
\begin{align*}
N_p^{1/2}\left(\sum_{i=1}^nX_i\right) 
	&\le  
		N_p^{1/2}\left(nX_1\right) 
		=
		\sum_{i=1}^nN_p^{1/2}\left(X_i\right).
\end{align*}
The second inequality can be derived from Lemma~\ref{lem:comp-cvx}, 
combined with Theorem~\ref{thm:maxpres} and the 
monotonicity of R\'enyi entropies:
\begin{align*}
	N_p^{1/2}\left(\sum_{i=1}^nX_i\right)
		&\leq
			N_{s_{\kappa,d} +1}^{1/2}\left(\sum_{i=1}^nX_i\right)
			\\
		&\leq 
			N_{s_{\kappa,d}+1}^{1/2}\left(nX_i\right)
			\\
		&=
			\exp \left( h_{s_{\kappa,d}+1}(X_i)/d + \log n \right)
			\\
		&\leq
			\exp \left( h_{p}(X_i)/d +  \left[\log n + \frac{\log(s_{\kappa,d}+1)}{s_{\kappa,d}} - \frac{\log p}{p-1} \right]\right)
			\\
		&=
			\frac{(s_{\kappa,d}+1)^{1/s_{\kappa,d}}}{p^{1/(p-1)}} \sum_{i=1}^n N^{1/2}_p(X_i).
\end{align*}
\end{proof}

Observe that Corollary~\ref{cor:ident-tri} is very reminiscent of Conjectures~\ref{conj:bnt1} and \ref{conj:gen-bnt};
the main difference is that here we have the assumption of identical marginals as opposed to central symmetry of the
joint distribution. 

We state the next corollary as a direct application of Corollary \ref{cor:ident-tri} for the log-concave case. 

\begin{cor}\label{cor:indp}
Suppose $X_1$, $X_2$, $\cdots$, $X_n$ are identically distributed log-concave random vectors in $\R^d$. Then
\be\label{eqn:repii.i.d.}
&&N_p\left(\sum_{i=1}^nX_i\right)\le n^2N_p(X_1)~\text{for}~p\in[0,1] ,\\
&&N_p\left(\sum_{i=1}^nX_i\right)\le e^2p^{2/(1-p)}n^2N_p(X_1)\le e^2n^2N_p(X_1)~\text{for}~p\in(1,\infty] .
\ee
In particular, if $X$ and $X'$ are identically distributed log-concave random vectors, then
\ben
&&N_p(X+X')\le 4N_p(X)~\text{for}~p\in[0,1],\\
&&N_p(X+X')\le 4e^2p^{2/(1-p)}N_p(X)\le 4e^2N_p(X)~\text{for}~p\in(1,\infty] .
\een
\end{cor}

Cover and Zhang \cite{CZ94} proved the remarkable fact that if
$X$ and $X'$ (possibly dependent) have the same log-concave distribution on $\RL$,
then
$h(X+X') \leq h(2X)$ 
(in fact, they also showed a converse of this fact). As observed by \cite{MK15}, their method also works
in the multivariate setting, where it implies 
that $N(X+X')\leq 4 N(X)$ for real-valued, i.i.d. log-concave $X, X'$.
This fact is recovered by the previous corollary.

Let us finally remark that if we are not interested in an explicit constant, then a version of this inequality
already follows from the Reverse EPI of \cite{BM12:jfa}. 
Indeed, 
\ben
N(X+X')\leq C N(X) ,
\een
since the same unit-determinant affine transformation must put both $X$ and $X'$
in $M$-position. However, the advantage of the methods we have explored is that we can obtain
explicit constants.

\subsection{Remarks on special positions that yield reverse EPI's}
\label{ss:spl-pos}

Let us recall the definition of isotropic bodies and measures in the convex geometric sense.  

\begin{defn} \normalfont
A convex body $K$ in $\mathbb{R}^d$ is called isotropic if there exists a constant $L_K$ such that
\[
\frac{1}{|K|^{1+ \frac 2 d}} \int_K \langle x, \theta \rangle^2 dx  = L^2_K,
\]
for all unit vectors $\theta\in \mathbb{S}^{d-1}$. More generally, a probability measure $\mu$ on $\R^d$
is called isotropic if there exists a constant $L_K$ such that
\[
\int_{\R^d} \langle x, \theta \rangle^2 \mu(dx)  = L^2_K,
\]
for all unit vectors $\theta\in \mathbb{S}^{d-1}$.
\end{defn}

The notion of $M$-position (i.e., a position or choice of affine transformation applied to convex bodies for which
a reverse Brunn-Minkowski inequality holds) was first introduced by V.~Milman \cite{Mil86}. Alternative approaches
to proving the existence of such a position were developed in \cite{Mil88:1, Pis89:book, GPV14}.
It was shown by Bobkov \cite{Bob11} that if the standard Gaussian measure conditioned to lie in a convex body $K$
is isotropic, then the body is in $M$-position and the reverse BMI applies. The notion of $M$-position was
extended from convex bodies  to log-concave measures in \cite{BM11:cras}, and further to convex measures in \cite{BM12:jfa}.
Using this extension, together with the sufficient condition obtained in \cite{Bob11}, one can give an explicit description 
of a position for which a reverse EPI applies with a universal-- but not explicit-- constant.

Nonetheless there are other explicit positions for which one can get reverse EPI's with explicit (but not dimension-independent) constants.
One instance of such is obtained from an extension to convex measures obtained by Bobkov \cite{Bob10} for
Hensley's theorem (which had earlier been extended from convex sets to log-concave functions by Ball \cite{Bal86}).

\begin{thm} \cite{Bob10} \label{thm:bobhens}
	For a symmetric, convex probability measure $\mu$ on $\mathbb{R}^d$ with density $f$ such that the body $\Lambda_f^{d-k}$ is isotropic, we have for any linear two subspaces $H_1$, $H_2$ of codimension $k$,
	\begin{align*}
		\int_{H_1} f dx \leq C_k \int_{H_2} f dx.
	\end{align*}
	What is more, $C_k < \left(\frac 1 2 e^2 \pi k\right)^{\frac k 2}$.
\end{thm}

As a consequence we have the following reverse $\infty$-R\'enyi EPI in the isotropic context.

\begin{cor}
Suppose the  joint distribution of the random vector $(X,Y)\in \R^d\times\R^d$ is symmetric and convex, 
with density $f=f(x,y)$. If the body $\Lambda_f^d$ is isotropic, then 
\ben
N_\infty(X+Y) \leq  \pi e^2 d \min\{ N_\infty(X), N_\infty(Y) \} .
\een
\end{cor}

\begin{proof}
Define two $d$-dimensional subspaces of $\mathbb{R}^{d}$: $H_1:=\{x=0\}$, $H_2:=\{x+y=0\}$.  Computing directly and applying Theorem \ref{thm:bobhens} we have our result as follows,
\begin{align*}
	\frac{N_\infty(X+Y)}{N_\infty(X)} 
		&= 
			\left(\frac{\|f_X \|_\infty}{\|f_{X+Y}\|_\infty} \right)^{\frac 2 d}
			\\
		&=
			\left(\frac{\int_{\mathbb{R}^d} f(0,z) dz}{\int_{\mathbb{R}^d} f(z,-z) dz } \right)^{\frac 2 d}
			\\
		&=
			\left( \frac{ 2^{\frac d 2} \int_{H_1} f}{\int_{H_2} f } \right)^{\frac 2 d}
			\\
		&\leq
			\pi e^2 d.
\end{align*}
\end{proof}

\section{The relationship between functional and entropic liftings}
\label{sec:reln}

In this section, we observe that the integral lifting of an inequality in Convex Geometry
may sometimes be seen as a R\'enyi entropic lifting.  

We start by considering integral and entropic liftings of a classical inequality
in Convex Geometry, namely the Blaschke-Santal\'o inequality.
For a convex body  $K\subset \R^d$ with $0 \in \text{int}(K)$, the polar $K^\circ$ of $K$ is defined as
\ben
K^\circ= \{y \in \R^d : \langle x, y\rangle \leq 1 \text{ for all } x \in K \},
\een
and, more generally, the polar $K^z$ with respect to $z \in \text{int}(K)$ by $(K-z)^\circ$. 
There is a unique point $s \in \text{int}(K)$, called the Santal\'o point of $K$, such that 
the volume product $|K| |K^s|$ is minimal-- it turns out that this point is such that the barycenter
of $K^s$ is 0. The Blaschke-Santal\'o inequality states that 
\ben
|K| |K^s|\leq |B_2^d|^2 ,
\een
with equality if and only if K is an ellipsoid.
In particular, the volume product $|K| |K^\circ|$ of a centrally symmetric convex body $K$ is maximized by the Euclidean ball.
This inequality was proved by Blaschke \cite{Bla17} in dimensions 2 and 3, and by Santal\'o \cite{San49} in general dimension;
the equality conditions were settled by Petty \cite{Pet85}. There have been many subsequent proofs;
see \cite{BK15} for a recent Fourier analytic proof as well as a discussion of the earlier literature.

More generally, if  $K, L$ are compact sets in $\R^d$, then
\be\label{eq:bs}
|K|\cdot |L| \leq \omega_d^2 \max_{x\in K, y\in L} |\langle x, y\rangle|^d.
\ee
The inequality \eqref{eq:bs} implies the Blaschke-Santal\'o inequality
by taking $K$ to be a symmetric convex body, and $L$ to be the polar of $K$.

Let us now describe an integral lifting of the inequality \eqref{eq:bs}, which was proved
by Lehec \cite{Leh09:1, Leh09:2} building on earlier work of Ball \cite{Bal88}, 
Artstein-Klartag-Milman \cite{AKM04}, and Fradelizi-Meyer \cite{FM07}. 

Let $f$ and $g$ be non-negative Borel functions on $\R^d$ satisfying the duality relation
\ben
\forall x,y \in\R^d,~ f(x)g(y)\le e^{-\langle x,y\rangle}.
\een
If $f$ (or $g$) has its barycenter (defined as $(\int f)^{-1}\int xf(x)dx$) at 0 then
\ben
\int_{\R^d}f(x)dx\int_{\R^d}g(y)dy\le  (2\pi)^d.
\een

The inequality \eqref{eq:bs} also has an entropic lifting.
For any two independent random vectors $X$ and $Y$ in $\R^d$, Lutwak-Yang-Zhang \cite{LYZ04}
showed that 
\be\label{eq:bs-ent}
N(X) \cdot N(Y) \leq   \frac{4\pi^2e^2}{d} \E \big[ |\langle X, Y\rangle |^2 \big] ,
\ee
with equality achieved for Gaussians. 
They also have an even more general (and still sharp) statement 
that bounds $[N_p(X) N_p(Y)]^{p/2}$ in terms of $\E [ |\langle X, Y\rangle |^p ]$,
with extremizers being certain generalized Gaussian distributions. 
As $p\ra\infty$, the expression $(\E [ |\langle X, Y\rangle |^p ])^{1/p}$ approaches
the essential supremum of $|\langle X, Y\rangle |$, which in the case that $X$ and $Y$ are
uniformly distributed on convex bodies is just the maximum that appears in the right side of 
inequality~\eqref{eq:bs}. Thus the Blaschke-Santal\'o inequality appears as the $L_\infty$
instance of the family of inequalities proved by Lutwak-Yang-Zhang \cite{LYZ04},
whereas the entropic lifting \eqref{eq:bs-ent} is the $L_2$ instance of the same family.
This perspective of entropy inequalities as being tied to an $L_2$-analogue of the
Brunn-Minkowski theory is greatly developed in a series of papers by Lutwak, Yang, Zhang,
sometimes with additional coauthors (see, e.g., \cite{LLYZ13} and references therein),
but this is beyond the scope of this survey.

For a function $V: \mathbb{R}^d \to \mathbb{R}$, its Legendre transform $\mathcal{L}V$ is defined by
\ben 
    \mathcal{L}V(x) = \sup_y \left[ \langle x , y\rangle - V(y) \right].
\een
For $f = e^{-V}$ log-concave, following Klartag and V.~Milman \cite{KM05}, we define its polar by
\ben    
    f^\circ = e^{- \mathcal{L}V}.
\een
Some basic properties of the polar are collected below.

\begin{lem}\label{lem:polar}
Let $f$ be a non-negative function on $\R^d$.
\begin{enumerate}
\item If $f$ is log-concave, then
\be\label{eq:Lprop1}
(f^\circ)^\circ=f.
\ee
\item If $g$ is also a non-negative function on $\R^d$, and the ``Asplund product'' of $f$ and $g$ is defined by
$f\star g(x)= \sup_{x_1+x_2=x} f(x_1) g(x_2)$, then
\be\label{eq:Lprop2}
(f\star g)^\circ=f^\circ g^\circ.
\ee
\item For any linear map $u$: $\R^d\rightarrow \R^d$ with full rank, we have the composition identity
\be\label{eq:Lprop3}
f^\circ\circ u=\left(f\circ u^{-T}\right)^\circ ,
\ee 
where $u^{-T}$ is the inverse of the adjoint of $u$.
\item If $f(x)$ takes its maximum value at $x=0$, one has 
\be\label{eq:Lprop4}
\sup f^\circ=\frac{1}{\sup f}~.
\ee
\item For any $p>0$, 
\be\label{eq:Lprop5}
(f^\circ)^p(x)=(f^p)^\circ(p x).
\ee
\end{enumerate}
\end{lem}

\begin{proof}
Write $f:=e^{-V}$ for a function $V:\R^d\ra\R$. 
The first two properties  are left as an exercise for the reader-- these are also
standard facts about the Legendre transform and its relation to the infimal convolution
of convex functions (see, e.g., \cite{Roc97:book}). For the third, 
we have
\ben
\left(f^\circ\circ u\right)(x)
&=&e^{-\sup_y \left[\langle ux, y\rangle-V(y)\right]  }
=e^{-\sup_y \left[\langle x, u^T y\rangle-V(y)\right]  }\\
&=&e^{-\sup_y \left[\langle x, y\rangle-V(u^{-T}y)\right]  }
=\left(f\circ u^{-T}\right)^\circ(x),
\een
which proves the property.

For the fourth, observe that we have, for any $x\in\R^d$,
\ben
\mathcal{L}V(x)=\sup_y\left[\langle x,y\rangle-V(y)\right]\ge -V(0).
\een
On the other hand, 
\ben
\mathcal{L}V(0)=\sup_y \left[-V(y)\right]= -V(0) .
\een
Thus we have  proved that
$\inf \mathcal{L}V=-V(0)$, 
which is equivalent to the desired property.

The last property is checked by writing
$(f^\circ)^p(x)=e^{-\sup_y\left[\langle px,y\rangle-pV(y)\right]}$.
\end{proof}

Bourgain and V.~Milman \cite{BM87} proved a {\it reverse} form of the Blaschke-Santal\'o
inequality, which asserts that there is a universal positive constant $c$ such that
\ben
|K| \cdot |K^\circ| \geq c^d ,
\een
for any symmetric convex body $K$ in $\R^d$, for any dimension $d$.
Klartag and V.~Milman \cite{KM05} obtained a functional lifting of this reverse inequality.

\begin{thm}\cite{KM05}\label{thm:km-bs}
There exists a universal constant $c>0$ such that for any dimension $d$ 
and for any  log-concave function $f: \mathbb{R}^d \to [0,\infty)$ centered at the origin 
(in the sense that $f(0)$ is the maximum value of $f$) with $0 < \int_{\mathbb{R}^d} f < \infty$, 
\ben
c^d < \left( \int_{\mathbb{R}^d} f \right) \left( \int_{\mathbb{R}^d} f^\circ \right) < (2\pi)^d.
\een
\end{thm}

Note that the upper bound here is just a special case of the integral lifting of the Blaschke-Santal\'o
inequality discussed earlier.

We observe that Theorem~\ref{thm:km-bs} can be thought of in information-theoretic terms, 
namely as a type of certainty/uncertainty principle.

\begin{thm}Let $X\sim f$ be a log-concave random vector in $\mathbb{R}^d$,
which is centered at the origin in the sense that $f$ is maximized there. 
Let $X^\circ$ be a random vector in $\R^d$ drawn from the density $f^\circ/\int_{\R^d} f^\circ$. 
Define the constants 
\begin{align*}
A_{p,d}&:=\frac{d(\log 2\pi-\log p-p\log c)}{1-p} ,\\
B_{p,d}&:=\frac{d(\log c-\log p-p\log 2\pi)}{1-p} ,
\end{align*}
where the constant $c$ is the same as in Theorem~\ref{thm:km-bs}. 
Then, for $p> 1$, we have
\be\label{inq:cucp1}
\max\{d\log c, A_{p,d}\} \,\le\, h_p(X)+h_p(X^\circ) \,\le\, \min \{d(\log 2\pi+2), B_{p,d}\},
\ee
and for $p<1$, we have
\be\label{inq:cucp2}
\max\{d\log c, B_{p,d}\} \,\le\, h_p(X)+h_p(X^\circ) \,\le\, \min\left\{d\left(\frac{2\log p}{p-1}+\log2\pi \right), A_{p,d}\right\}~.
\ee
In particular, if $p=\infty$,
\be\label{inq:cucpinfty}
d\log c\le h_\infty(X)+h_\infty(X^\circ)\le d\log 2\pi ,
\ee
and for $p=1$,
\begin{equation} \label{certainty}
d\log c \le  h(X) + h(X^\circ) \le  d(\log 2\pi+2).
\end{equation}
\end{thm}

\begin{proof}
We have 
\be\label{eq:sumRenyi}
h_p(X)+h_p(X^\circ)=\frac{\log\left[\int f^p \int (f^\circ)^p\right]-p\log \int f^\circ}{1-p}~.
\ee
By property \eqref{eq:Lprop5}, we have $\int (f^\circ)^p=\frac{1}{p^d}\int (f^p)^\circ$. So by \eqref{eq:sumRenyi}:
\ben
h_p(X)+h_p(X^\circ)=\frac{\log\left[\int f^p \int (f^p)^\circ\right]-d\log p-p\log \int f^\circ}{1-p}~.
\een
Thus, by applying Theorem \ref{thm:km-bs} twice, if $p>1$:
\ben
h_p(X)+h_p(X^\circ)\ge \frac{d\log 2\pi-d\log p-p\log \int f^\circ}{1-p}\ge A_{p,d}.
\een
On the other hand, 
\ben
h_p(X)+h_p(X^\circ)\le  \frac{d\log c-d\log p-p\log \int f^\circ}{1-p}\le B_{p,d}.
\een
Therefore we have
\be\label{inq:cucp1'}
A_{p,d}\le h_p(X)+h_p(X^\circ)\le B_{p,d}.
\ee
A similar argument for $p<1$ gives
\be\label{inq:cucp2'}
B_{p,d} \le h_p(X)+h_p(X^\circ)\le A_{p,d}.
\ee
Letting $p\rightarrow\infty$, we have \eqref{inq:cucpinfty}. 
For $p=1$, by Lemma \ref{lem:theLemma} and \eqref{inq:cucpinfty}, 
\ben
n\log c\le  h_\infty(X)+h_\infty(X^\circ)\le  h(X) + h(X^\circ)\le h_\infty(X)+h_\infty(X^\circ)+2n\le n(\log 2\pi+2),
\een
which provides \eqref{certainty}. Thus for $p>1$, by \eqref{inq:cucpinfty}, \eqref{certainty} and Lemma~\ref{lem:monotonicity}, we also have
\ben
n\log c\le  h_\infty(X)+h_\infty(X^\circ)\le  h_p(X) + h_p(X^\circ)\le  h(X) + h(X^\circ)\le n(\log 2\pi+2).
\een
Combining with \eqref{inq:cucp1'} provides \eqref{inq:cucp1}.
which provides the theorem. For $p<1$, we have, by \eqref{certainty} and Lemma~\ref{lem:monotonicity}, we have 
\ben
d\log c \le  h(X) + h(X^\circ) \le h_p(X) + h_p(X^\circ).
\een
Combining this with \eqref{inq:cucp2'} provides the left most inequality of \eqref{inq:cucp2}. And by applying Lemma~\ref{lem:theLemma} on $h_p(f)-h_\infty(f)$ and by \eqref{inq:cucpinfty}, we have
\ben
h_p(X) + h_p(X^\circ)\le \frac{2d\log p}{p-1}+h_\infty(X) + h_\infty(X^\circ)\le \frac{2d\log p}{p-1}+d\log 2\pi.
\een
Combining this with \eqref{inq:cucp2'} gives \eqref{inq:cucp2}.
\end{proof}

Klartag and Milman \cite{KM05} prove a reverse Pr\'ekopa-Leindler inequality (Reverse PLI).

\begin{thm}\cite{KM05}\label{thm:revPLI}
Given $f,g$: $\R^d\rightarrow [0,\infty)$ be even log-concave functions with $f(0)=g(0)=1$, then there exist $u_f$, $u_g$ in $SL(d)$ such that $\bar{f} = f\circ u_f$, $\bar{g} = g\circ u_g$ satisfy 
\ben 
\left( \int \bar{f} \star \bar{g} \right)^{\frac 1 d} \leq C \left( \left(\int \bar{f} \right)^{\frac 1 d} + \left( \int \bar{g} \right)^{\frac 1 d}  \right),
\een
where $C > 0$ is a universal constant, $u_f$ depends solely on $f$, and $u_g$ depends solely on $g$.
\end{thm}

We observe that the Reverse PLI can be proved from the Positional Reverse R\'enyi EPI we proved earlier, modulo
the reverse functional Blaschke-Santal\'o inequality of Klartag-Milman. 

\begin{prop}
Theorems \ref{thm:REPIREnyi} and \ref{thm:km-bs} together imply Theorem \ref{thm:revPLI}.
\end{prop}

\begin{proof}
Let $f$, $g$: $\R^d\rightarrow [0,\infty)$ be even log-concave functions with $f(0)=g(0)=1$. Now by property \eqref{eq:Lprop4}, $\|f^\circ\|_\infty=1$ as well. Now apply reversed $\infty$-EPI on a pair of independent random vectors $X$ and $Y$ with density functions $f^\circ/\int f^\circ$ and $g^\circ/\int g^\circ$ respectively, there exist linear maps $u_1$, $u_2\in SL(d)$ depending solely on $f$ and $g$ respectively, such that
\begin{align*}
&\left(\int \frac{\left(f^\circ \circ u_1(x)\right)\cdot \left(g^\circ \circ u_2(x)\right)}{\int f^\circ\cdot \int g^\circ}dx \right)^{-\frac{2}{d}} 
= N_\infty(u_1(X)+u_2(Y))\\
\lesssim & N_\infty(X) +N_\infty(Y)
= \left\|\frac{f^\circ}{\int f^\circ}\right\|_\infty^{-\frac{2}{d}}+\left\|\frac{g^\circ}{\int g^\circ}\right\|_\infty^{-\frac{2}{d}}=\left(\int f^\circ\right)^{\frac{2}{d}}+\left(\int g^\circ\right)^{\frac{2}{d}}.
\end{align*}
Therefore we have
\be\label{inq:middleStep}
\left(\int \left(f^\circ \circ u_1(x)\right)\cdot \left(g^\circ \circ u_2(x)\right)dx \right)^{-\frac{2}{d}} \lesssim \left(\int f^\circ\right)^{-\frac{2}{d}}+\left(\int g^\circ\right)^{-\frac{2}{d}}.
\ee
Thus by Theorem~\ref{thm:km-bs}, we have the right hand side of \eqref{inq:middleStep} is
\be\label{inq:middleStep2}
\left(\int f^\circ\right)^{-\frac{2}{d}}+\left(\int g^\circ\right)^{-\frac{2}{d}}\lesssim \left(\int f\right)^{\frac{2}{d}}+\left(\int g\right)^{\frac{2}{d}}.
\ee
On the other hand, by properties \eqref{eq:Lprop1}, \eqref{eq:Lprop2} and \eqref{eq:Lprop3}, we have the right hand side of \eqref{inq:middleStep}:
\be\label{inq:middleStep3}
\left(\int \left(f^\circ \circ u_1(x)\right)\cdot \left(g^\circ \circ u_2(x)\right)dx \right)^{-\frac{2}{d}}\gtrsim
\left(\int \left(f \circ u_1^{-t}\right)\star \left(g \circ u_2^{-t}\right) \right)^{\frac{2}{d}}.
\ee
Denote $u_f:=u_1^{-t}$, $u_g:=u_2^{-t}$; $\bar{f} := f\circ u_f$, $\bar{g} := g\circ u_g$, and combining \eqref{inq:middleStep} \eqref{inq:middleStep2} and \eqref{inq:middleStep3} provides Theorem \ref{thm:revPLI}.
\end{proof}

\section{Concluding remarks}
\label{sec:concl}

One productive point of view put forward by Lutwak, Yang and Zhang is that the correct
analogy is between entropy inequalities and the inequalities of the $L^2$-Brunn-Minkowski
theory rather than the standard Brunn-Minkowski theory. While we did not have space to
pursue this direction in our survey apart from a brief discussion in Section~\ref{sec:reln}, 
we refer to \cite{LLYZ13} and references therein for details. 

A central question when considering integral or entropic liftings of Convex Geometry is whether there exist integral and entropic analogues
of mixed volumes. Recent work of Bobkov-Colesanti-Fragala \cite{BCF14} has shown that an integral lifting of
intrinsic volumes does exist, and Milman-Rotem \cite{MR13:1, MR13:2} independently showed this as well as 
an integral lifting of mixed volumes more generally. A fully satisfactory theory of 
``intrinsic entropies'' or ``mixed entropies'' is yet to emerge, 
although some promising preliminary results in this vein can be found in \cite{JA15:isit}. 

It is also natural to explore nonlinear generalizations, to ambient spaces that are manifolds or groups.
Log-concave (and convex) measures can be put into an even broader context by viewing them
as instances of curvature in metric measure spaces. 
Indeed, thanks to path-breaking work of \cite{Stu06:1, LV09}, 
it was realized that one can give meaning (synthetically) to the notion of a lower bound on Ricci curvature for a metric space
$(\mathcal{X}, d)$ equipped with a measure $\mu$ (thus allowing for geometry beyond the traditional setting of Riemannian manifolds).
In particular, they extended the celebrated Curvature-Dimension condition $CD(K, N)$ of Bakry and \'Emery \cite{BE85} 
to metric measure spaces $(\mathcal{X}, d, \mu)$;
the simplest case $CD(K,\infty)$ is defined by a ``displacement convexity'' (or convexity along 
optimal transport paths) property of the relative entropy functional $D(\cdot\|\mu)$.
For Riemannian manifolds, the $CD(K, N)$ condition is satisfied if and only if the manifold has dimension at most $N$ 
and Ricci curvature at least $K$, while Euclidean space $\RL^d$ equipped with a log-concave measure
may be thought of as having non-negative Ricci curvature in the sense that it satisfies $CD(0, d)$.
Moreover,  $\RL^d$ equipped with a convex measure may be 
interpreted as a $CD(K, N)$ space with effective dimension $N$ being negative (other examples can be found in \cite{Mil15}).
In these more general settings (where there may not be a group structure), it is not entirely clear whether there are natural formulations of entropy power
inequalities. Even for the case of Lie groups, almost nothing seems to be known.

One may also seek discrete analogs of the phenomena studied in this survey, which are closely
related to investigations in additive combinatorics. In discrete settings, additive structure plays
a role as or more important than that of convexity. The Cauchy-Davenport inequality is an analog
of the Brunn-Minkowski inequality in cyclic groups of prime or infinite order, with arithmetic progressions 
being the extremal objects (see, e.g., \cite{TV06:book}); extensions to the integer lattice are also
known \cite{Ruz94, GG01, Sta01}. A probabilistic lifting of the Cauchy-Davenport inequality for the integers 
is presented in \cite{WWM14:isit}. Sharp lower bounds on entropies of sums in terms of those of summands
are still not known for most countable groups; partial results in this direction may be found in 
\cite{Tao10, HAT14, JA14, WM15:isit}. 
Such bounds are also relevant to the study of information-theoretic
approaches to discrete limit theorems, such as those that involve distributional convergence
to the Poisson or compound Poisson distributions of sums of random variables 
taking values in the nonnegative integers; we refer the interested reader to
\cite{Joh07, JKM13, Yu09:1, Yu09:2, BJKM10} for further details.
Probabilistic liftings of other ``sumset inequalities'' from
additive combinatorics can be found in \cite{Mad08:itw, Ruz09, MMT10:itw, MMT12, Tao10, MK10:isit, ALM15, MK15, LM16}.

There are other connections between notions of entropy and convex geometry that we have not discussed
in this paper (see, e.g., \cite{BM11:aop, AKSW12, Wer12, CW14, FMW16, FLM15, LFM16:isit}).

\section*{Acknowledgement}

The authors are grateful to Eric Carlen, Bernardo Gonz\'alez Merino, Igal Sason, Tomasz Tkocz, Elisabeth Werner, and an anonymous reviewer
for useful comments and references.


\end{document}